%% file: main.tex
\documentclass[11pt]{article}

\usepackage[letterpaper,margin=1in]{geometry}
\usepackage[T1]{fontenc}
\usepackage[utf8]{inputenc}
\usepackage{lmodern}
\usepackage{microtype}

\usepackage{amsmath,amssymb,amsthm,mathtools,bm}
\usepackage{bbm}
\usepackage{dsfont}

\usepackage{graphicx}
\usepackage{xcolor}
\usepackage{float}
\usepackage{enumitem}
\usepackage{booktabs}
\usepackage{array}
\usepackage{multirow}
\usepackage{algorithm}
\usepackage{algpseudocode}
\usepackage{booktabs}
\usepackage{placeins}
\usepackage{tabularx}
\usepackage{array}
\usepackage{comment}

\usepackage[hidelinks]{hyperref}
\usepackage[nameinlink,capitalise,noabbrev]{cleveref}

\usepackage{tikz}
\usetikzlibrary{arrows.meta,positioning,calc,decorations.pathreplacing}

\setlist[itemize]{topsep=2pt,itemsep=2pt,parsep=1pt,partopsep=1pt}
\setlist[enumerate]{topsep=2pt,itemsep=2pt,parsep=1pt,partopsep=1pt}

\usepackage{aliascnt}

\newtheorem{theorem}{Theorem}[section]

\newaliascnt{lemma}{theorem}
\newtheorem{lemma}[lemma]{Lemma}
\aliascntresetthe{lemma}

\newaliascnt{proposition}{theorem}

\aliascntresetthe{proposition}

\newaliascnt{corollary}{theorem}
\newtheorem{corollary}[corollary]{Corollary}
\aliascntresetthe{corollary}

\newaliascnt{claim}{theorem}

\aliascntresetthe{claim}

\newaliascnt{fact}{theorem}

\aliascntresetthe{fact}

\newaliascnt{observation}{theorem}

\aliascntresetthe{observation}

\theoremstyle{definition}

\newaliascnt{definition}{theorem}

\aliascntresetthe{definition}

\newaliascnt{example}{theorem}

\aliascntresetthe{example}

\newaliascnt{remark}{theorem}
\newtheorem{remark}[remark]{Remark}
\aliascntresetthe{remark}

\newaliascnt{openproblem}{theorem}
\newtheorem{openproblem}[openproblem]{Open Problem}
\aliascntresetthe{openproblem}

\crefname{theorem}{theorem}{theorems}
\Crefname{theorem}{Theorem}{Theorems}
\crefname{lemma}{lemma}{lemmas}
\Crefname{lemma}{Lemma}{Lemmas}
\crefname{proposition}{proposition}{propositions}
\Crefname{proposition}{Proposition}{Propositions}
\crefname{corollary}{corollary}{corollaries}
\Crefname{corollary}{Corollary}{Corollaries}
\crefname{claim}{claim}{claims}
\Crefname{claim}{Claim}{Claims}
\crefname{fact}{fact}{facts}
\Crefname{fact}{Fact}{Facts}
\crefname{observation}{observation}{observations}
\Crefname{observation}{Observation}{Observations}
\crefname{definition}{definition}{definitions}
\Crefname{definition}{Definition}{Definitions}
\crefname{example}{example}{examples}
\Crefname{example}{Example}{Examples}
\crefname{remark}{remark}{remarks}
\Crefname{remark}{Remark}{Remarks}
\crefname{algorithm}{algorithm}{algorithms}
\Crefname{algorithm}{Algorithm}{Algorithms}
\crefname{section}{section}{sections}
\Crefname{section}{Section}{Sections}

\newcommand{\eps}{\epsilon}

\newcommand{\R}{\mathbb{R}}

\newcommand{\dF}{d_{\mathrm F}}
\newcommand{\ddF}{d_{\mathrm{dF}}}

\newcommand{\REACH}{\textnormal{\textsc{Reach}}}

\newcommand{\defeq}{\vcentcolon=}
\newcommand{\abs}[1]{\left|#1\right|}
\newcommand{\norm}[1]{\left\lVert#1\right\rVert}
\newcommand{\set}[1]{\left\{#1\right\}}

\newcommand{\floor}[1]{\left\lfloor#1\right\rfloor}

\title{\texorpdfstring{$(5+\eps)$}{(5+epsilon)}-Approximation of Fr\'echet Distance in Strongly Subquadratic Time\thanks{University of Illinois Urbana-Champaign, Urbana, IL 61801. Email: \texttt{\{hengyu2, jihanw2\}@illinois.edu}.}}
\author{Lenny Liu \qquad Jihan Wang}
\date{\today}

\begin{document}

\renewcommand{\thefootnote}{\fnsymbol{footnote}}
\maketitle
\thispagestyle{empty}
\renewcommand{\thefootnote}{\arabic{footnote}}
\setcounter{footnote}{0}

\begin{abstract}
We give randomized
$(5+\epsilon)$-approximation algorithms for both the continuous and discrete
Fr\'echet distances on arbitrary two polygonal curves $\tau$ and $\sigma$ in $\mathbb R^d$ for fixed $d$,
with $n$ and $m\le n$ vertices respectively.
Our algorithm for continuous Fr\'echet runs in $\widetilde O_{d,\epsilon}(n m^{8/9})$ time,
and our algorithm for discrete Fr\'echet runs in $\widetilde O_{d,\epsilon}(n m^{4/5})$
time.
These bounds improve the recent strongly subquadratic constant-factor
approximation algorithms of Cheng, Huang, and
Zhang~\cite{cheng2025constant}, which give $(7+\epsilon)$-approximations.

The approximation improvement comes from certifying long boundary-to-boundary
reachability directly through auxiliary surrogate curves, avoiding an extra
conversion back to input subcurves and hence removing one triangle-inequality
loss.  The running-time improvement comes from a two-scale macro-surrogate
search combined with dyadic auxiliary-transfer structures, with the discrete
case gaining a faster bound from exact planar reachability in the discrete
free-space graph.
\end{abstract}
\thispagestyle{empty}
\clearpage

\setcounter{page}{1}
\input{introduction}

\input{overview}

\input{preliminaries}

\input{continuous}

\input{discrete}

\input{Discussion}

\paragraph{Independent and concurrent work.}
Independently of and concurrently with this work, van~der~Horst and
Ophelders~\cite{vdHO26} obtained a deterministic
$(3+\epsilon)$-approximation for the continuous Fr\'echet distance in
$O\bigl(nm^{2/3}\log n \cdot \log(\tfrac{1}{\epsilon}\log n)\bigr)$ time,
applicable in general metric spaces, with an exact factor-$3$ algorithm in
$\mathbb{R}$; they state that analogous results hold for the discrete
Fr\'echet distance, with details deferred to their full version. Their
results supersede ours in approximation factor and running time. Both works were
carried out independently; we learned of theirs through personal
communication after this manuscript was completed. Their algorithm
propagates reachability exactly along surrogates formed by subcurves of
the input curve itself, whereas ours transfers reachability through
randomly sampled auxiliary curves. Their construction achieves, in effect,
the goal of \Cref{op:search}.

\bibliographystyle{alpha}
\bibliography{ref}

\end{document}

%% file: introduction.tex
%======================================================================
\section{Introduction}
\label{sec:intro}
%======================================================================

\emph{Fr\'echet distance} is a classical measure of similarity between curves.
In the standard ``man and dog'' interpretation, two agents traverse their
respective curves from start to finish without backtracking, and the distance
is the minimum leash length that permits such a traversal.  Because the
definition respects the order of points along each curve, the Fr\'echet distance
is well suited to comparing trajectories, paths, and time-series data.

Let $\tau,\sigma:[0,1]\to\R^d$ be polygonal curves.  A continuous matching
between $\tau$ and $\sigma$ is a pair of continuous nondecreasing maps
$\rho,\varrho:[0,1]\to[0,1]$ satisfying
$\rho(0)=\varrho(0)=0$ and $\rho(1)=\varrho(1)=1$.  The \emph{continuous}
Fr\'echet distance is
\[
    \dF(\tau,\sigma)
    =
    \min_{\rho,\varrho}
    \max_{t\in[0,1]}
    \norm{\tau(\rho(t))-\sigma(\varrho(t))}.
\]

For the discrete variant, write the vertex sequences as
$\tau=(v_1,\ldots,v_n)$ and $\sigma=(w_1,\ldots,w_m)$.  A
\emph{discrete matching} is a sequence
\[
    C=((i_1,j_1),\ldots,(i_L,j_L))
\]
such that $(i_1,j_1)=(1,1)$, $(i_L,j_L)=(n,m)$, and for every
$\ell<L$,
\[
    (i_{\ell+1}-i_\ell,\ j_{\ell+1}-j_\ell)
    \in \{(1,0),(0,1),(1,1)\}.
\]
The \emph{discrete} Fr\'echet distance is
\[
    \ddF(\tau,\sigma)
    =
    \min_C
    \max_{(i,j)\in C}
    \norm{v_i-w_j}.
\]

\paragraph{Exact algorithms.}
Alt and Godau~\cite{alt1995computing} initiated the algorithmic study of the
continuous Fr\'echet distance and gave the classical algorithm,
which computes the distance between curves of complexities $m$ and $n$ in
$O(mn\log(mn))$ time.  Eiter and Mannila~\cite{eiter1994computing} introduced
the discrete Fr\'echet distance and gave an $O(mn)$-time dynamic programming algorithm.
These quadratic-type bounds remain the natural baseline for arbitrary input
curves.

Several works have obtained subpolynomial improvements for exact computation.
For the discrete Fr\'echet distance in the plane, Agarwal, Ben Avraham, Kaplan,
and Sharir~\cite{agarwal2014computing} gave a word-RAM algorithm running in
$O(mn\log\log n/\log n)$ time, assuming $m\le n$.
For the continuous Fr\'echet distance, Cheng
and Huang~\cite{cheng2025frechet} recently gave an exact randomized algorithm in
arbitrary fixed dimension with expected running time
$O(mn(\log\log n)^{2+\mu}\log n/\log^{1+\mu}m)$ for some constant
$\mu\in(0,1)$.

\paragraph{Hardness and lower bounds.}
Conditional lower bounds suggest that truly subquadratic exact algorithms are
unlikely.  Bringmann~\cite{bringmann2014walking}, assuming the
\emph{Strong Exponential Time Hypothesis} (SETH) of Impagliazzo and
Paturi~\cite{impagliazzo2001complexity}, ruled out strongly subquadratic time exact
algorithms for both the continuous and the discrete Fr\'echet distance,
already for curves in the plane.  His lower bound holds even for imbalanced
complexities, so, it is likely no algorithm runs in $O((nm)^{1-\gamma})$ time for any constant
$\gamma>0$.  The same work also rules out strongly subquadratic time $1.001$-approximation.  Buchin, Ophelders, and
Speckmann~\cite{buchin2019seth} raised the inapproximability threshold, under
SETH, no strongly subquadratic algorithm approximates the continuous or the
discrete Fr\'echet distance within a factor smaller than $3$, even for curves
in one dimension.  They also showed that reductions of this kind cannot
establish an inapproximability factor above $3$.  

\paragraph{Approximation algorithms.}
For structured inputs, strong approximation guarantees are known.  Driemel,
Har-Peled, and Wenk~\cite{driemel2010approximating} introduced
$c$-packed curves and gave a near-linear-time $(1+\eps)$-approximation for them.  Bringmann and K\"unnemann~\cite{bringmann2017improved} subsequently
improved the dependence on $\eps$, matching conditional lower bounds up to
lower-order factors.

For arbitrary curves, earlier strongly subquadratic time algorithms achieved
approximation factors that grow with the input size.  For the discrete
Fr\'echet distance, Bringmann and Mulzer~\cite{bringmann2016approximability}
gave an $O(\alpha)$-approximation in $O(n\log n+n^2/\alpha)$ time for two
$n$-vertex curves, and Chan and Rahmati~\cite{chan2018improved} improved
the running time to $O(n\log n+n^2/\alpha^2)$.  For the continuous Fr\'echet
distance, Colombe and Fox~\cite{colombe2021approximating} gave the first
strongly subquadratic time algorithm with a polynomial approximation.  For
two $n$-vertex curves in fixed dimension and
$\alpha\in[\sqrt n,n]$, their algorithm computes an
$O(\alpha)$-approximation in
$O((n^3/\alpha^2)\log n)$ time. Van der Horst, van Kreveld, Ophelders, and
Speckmann~\cite{van2023subquadratic} later gave an
$O(\alpha)$-approximation for curves of complexities $m\le n$ in fixed
dimension in $O((n+mn/\alpha)\log^3 n)$ time.  More recently, van der Horst,
van Kreveld, Ophelders, and Speckmann~\cite{van2024faster} improved the
arbitrary-dimensional running time to $O((n^2/\alpha)\log n)$ for
equal-complexity curves and gave a one-dimensional algorithm running in
$O((n^2/\alpha^3)\log^2 n)$ time.

A recent breakthrough of Cheng, Huang, and
Zhang~\cite{cheng2025constant} gave the first strongly subquadratic
constant-factor approximation algorithms for continuous and discrete
Fr\'echet distance on arbitrary curves.  For curves $\tau$ and $\sigma$ with
$\abs{\tau}=n$ and $\abs{\sigma}=m\le n$, they obtained randomized
$(7+\eps)$-approximation algorithms for both variants in
$\widetilde O_{d,\eps}(nm^{0.99})$ time.

In subsequent work focused on the imbalanced setting,
Blank~\cite{blank2026frechet} obtained a $(3+\eps)$-approximation for both
variants in $O((n+m^2)\log n)$ time and refined the corresponding conditional
lower bounds.  This running time is particularly effective when the smaller
curve is substantially shorter.

\begin{table*}[t]
    \centering
    \small
    \setlength{\tabcolsep}{4.5pt}
    \renewcommand{\arraystretch}{1.08}
    \begin{tabularx}{\textwidth}{@{}
    l
    >{\raggedright\arraybackslash}p{0.20\textwidth}
    >{\raggedright\arraybackslash}p{0.14\textwidth}
    >{\raggedright\arraybackslash}X
    l
@{}}
        \toprule
        Variant
        & Setting
        & Guarantee
        & Running time
        & Reference
        \\
        \midrule
        Continuous
        & arbitrary curves
        & exact
        & $O(mn\log(mn))$
        & \cite{alt1995computing}
        \\
        Discrete
        & arbitrary curves
        & exact
        & $O(mn)$
        & \cite{eiter1994computing}
        \\
        Discrete
        & $d=2$
        & exact
        & $O(mn\log\log n/\log n)$
        & \cite{agarwal2014computing}
        \\
        Continuous
        & fixed $d$
        & exact, randomized
        & $O\!\left(
            mn(\log\log n)^{2+\mu}\log n/
            \log^{1+\mu}m
          \right)$
        & \cite{cheng2025frechet}
        \\
        \midrule
        Discrete
        & equal complexity
        & $O(\alpha)$
        & $O(n\log n+n^2/\alpha)$
        & \cite{bringmann2016approximability}
        \\
        Discrete
        & equal complexity
        & $O(\alpha)$
        & $O(n\log n+n^2/\alpha^2)$
        & \cite{chan2018improved}
        \\
        Continuous
        & equal complexity, fixed $d$
        & $O(\alpha)$
        & $O((n^3/\alpha^2)\log n)$
        & \cite{colombe2021approximating}
        \\
        Continuous
        & $m\le n$, fixed $d$
        & $O(\alpha)$
        & $O((n+mn/\alpha)\log^3 n)$
        & \cite{van2023subquadratic}
        \\
        Continuous
        & equal complexity
        & $O(\alpha)$
        & $\widetilde O(n^2/\alpha)$
        & \cite{van2024faster}
        \\
        \midrule
        Continuous
        & arbitrary curves
        & $7+\eps$
        & $\widetilde O_{d,\eps}(nm^{0.99})$
        & \cite{cheng2025constant}
        \\
        Discrete
        & arbitrary curves
        & $7+\eps$
        & $\widetilde O_{d,\eps}(nm^{0.99})$
        & \cite{cheng2025constant}
        \\
        Continuous
        & $m\le n$; imbalanced focus
        & $3+\eps$
        & $O((n+m^2)\log n)$
        & \cite{blank2026frechet}
        \\
        Discrete
        & $m\le n$; imbalanced focus
        & $3+\eps$
        & $O((n+m^2)\log n)$
        & \cite{blank2026frechet}
        \\
        \midrule
        Continuous
        & arbitrary curves
        & $5+\eps$
        & $\widetilde O_{d,\eps}(nm^{8/9})$
        & \Cref{thm:continuous-main}
        \\
        Discrete
        & arbitrary curves
        & $5+\eps$
        & $\widetilde O_{d,\eps}(nm^{4/5})$
        & \Cref{thm:discrete-main} 
        \\
        \bottomrule
    \end{tabularx}
    \caption{
        Selected exact and approximation algorithms for Fr\'echet distance,
        where $m\le n$.  Unless stated otherwise, the dimension is fixed.
    }
    \label{tab:frechet-results}
\end{table*}

\paragraph{Our results.}
Cheng, Huang, and Zhang~\cite{cheng2025constant} identify improving both
the approximation factor and the running time of their strongly subquadratic
constant-factor algorithms as a natural question.  We answer this question in the
affirmative on both counts.  For arbitrary
polygonal curves $\tau$ and $\sigma$ of complexities $n$ and $m\le n$ in
fixed dimension, we give randomized $(5+\eps)$-approximation algorithms for
both continuous and discrete Fr\'echet distance.  The continuous Fr\'echet algorithm
runs in $\widetilde O_{d,\eps}(nm^{8/9})$ time, and the discrete one
runs in $\widetilde O_{d,\eps}(nm^{4/5})$ time.
Our fixed-threshold decision procedures distinguish distance at most
$\delta$ from larger distances within the same time bounds.  The
continuous procedure rejects distance greater than $(5+\eps)\delta$.  The
discrete procedure needs no accuracy parameter at all; it distinguishes
distance at most $\delta$ from distance greater than $5\delta$ in
$\widetilde O_d(nm^{4/5})$ time, independent of $\eps$. The parameter $\eps$ enters
the discrete approximation result only through the final threshold search, seeded by a coarse estimate of Bringmann and Mulzer~\cite{bringmann2016approximability}. 
The continuous approximation algorithm follows
from the decision-to-optimization transformation of Colombe and
Fox~\cite{colombe2021approximating}.

The rest of the paper is organized as follows.
\Cref{sec:overview} gives a technical overview of the algorithms.
\Cref{sec:preliminaries} introduces the free-space reachability framework,
curve simplification, and block decomposition.
\Cref{sec:continuous} proves the continuous gap-decision and approximation
results.
\Cref{sec:discrete} proves the corresponding discrete results.
\Cref{sec:discussion} poses an open problem.

%% file: overview.tex
%======================================================================
\section{Overview}
\label{sec:overview}

At a fixed threshold $\delta$, the standard continuous and discrete
Fr\'echet dynamic programs explore the full free-space diagram: the monotone 
reachability structure whose feasible cells or vertices encode pairs of positions at distance at most $\delta$.
The quadratic running time reflects an all-to-all
comparison: every position on $\tau$ is implicitly compared with
every position on $\sigma$.  Our strategy is to avoid these all-to-all comparisons by certifying reachability only at selected places, 
using two structural properties of Fr\'echet matchings
throughout: matchings \emph{restrict} (a matching of two curves
restricts to a matching between any subcurve of one and some subcurve
of the other), and matchings \emph{concatenate} at their maximum cost
rather than their sum.  Our approach has three components.  First, we partition 
both curves into blocks and reduce the computation to one difficult transition per block pair.  
Second, we certify that transition by transferring reachability through short auxiliary surrogate curves.  
Third, we find these surrogates by a two-scale macro search: random samples handle dense pieces, while a 
sparse fallback handles the remaining cases; all transfer queries are answered from data structures 
preprocessed per host block.

We adopt the block framework of Cheng, Huang, and
Zhang~\cite{cheng2025constant}.  Partition $\tau$ into consecutive
blocks of $\mu_1$ vertices and $\sigma$ into consecutive blocks of
$\mu_2$ vertices, for parameters fixed later, and process the grid of
block pairs $\tau_k\times\sigma_l$ in a topological order of the block grid.  
The algorithm stores information only on block boundaries: each boundary carries a
set that is \emph{complete}, containing every globally $\delta$-reachable point on that boundary, 
and \emph{$\beta$-sound}, meaning every stored point is globally
$\beta\delta$-reachable, where $\beta=5+\eps$ in the continuous case
(\Cref{sec:continuous}) and $\beta=5$ in the discrete case
(\Cref{sec:discrete}).  
One local update, \REACH{}, receives such sets on the left and bottom boundaries
of a block pair and must produce them on the right and top.
Completeness guarantees that a true matching is never discarded, and
$\beta$-soundness does not degrade as certificates cross many blocks, because
matchings concatenate at their maximum cost.  The global problem
therefore reduces to implementing a single \REACH{} update, and the
entire game is to do so in time far below the $\mu_1\mu_2$ cost of
solving one block pair exactly.

A witness path through a block pair enters from the left or bottom
boundary and exits through the right or top, giving four classes.
Three of them are easy for the same reason: each pins the host
side to a prefix, a suffix, or all of $\tau_k$, matched to a portion of
$\sigma_l$ with only $O(\mu_2)$ vertices.  A curve that matches an
$O(\mu_2)$-vertex curve at threshold $\delta$ admits a simplification
of $O(\mu_2)$ size, so the relevant host prefix, suffix, or block curve
can be replaced by its simplification and the witness re-certified by a
dynamic program of size $O(\mu_2^2)$.  The soundness of these local
cases stays below the final threshold, and their cost is lower order.
The fourth class is where the difficulty arises.  A bottom-to-top witness matches
the entire block $\sigma_l$ to an unknown subcurve
$P=\tau_k[x,y]$, and both endpoints are unknown: $x$ and $y$ range over
many scattered positions continuously many in the continuous setting.
Testing all source-endpoint pairs would recover exactly the all-to-all
comparison that the block decomposition was built to avoid.

The way out is that certifying such a witness does not require finding it.
Suppose the algorithm holds any short curve $\pi$ with
$d(\pi,\sigma')\le q\delta$
for the current query piece $\sigma'\subseteq\sigma_l$, where $d$ denotes
$\dF$ or $\ddF$ as appropriate; we call $\pi$ a \emph{surrogate} of quality
$q$.  A surrogate need not be a subcurve of $\tau_k$, or even of $\tau$.
An \emph{auxiliary transfer} takes a set $S$ of reachable sources on $\tau_k$,
the surrogate $\pi$, and a threshold, and returns the endpoints reachable from
$S$ by matching subcurves of $\tau_k$ to $\pi$: transfer completeness keeps every
true endpoint, and transfer soundness certifies every returned one.

The accounting is the same in the two variants.  Let $\lambda$ be the loss of
the transfer structure: $\lambda=1+\eta$ for the continuous portal-rounded
structure, where $\eta\le\eps/6$ is the internal accuracy parameter of the
continuous algorithm, and $\lambda=1$ for the discrete exact structure.  A true witness is a host subcurve $\tau_k[z,y]$ with
$d(\tau_k[z,y],\sigma')\le\delta$.  For every true witness, the triangle
inequality gives
$d(\tau_k[z,y],\pi)\le d(\tau_k[z,y],\sigma')+d(\sigma',\pi)\le(1+q)\delta$.
Transferring at threshold $(1+q)\delta$ therefore loses no witness.  Conversely, a
returned endpoint has a certificate against $\pi$ at threshold
$\lambda(1+q)\delta$, and converting it back to $\sigma'$ costs $q\delta$ once
more.  The resulting local soundness bound is
$\bigl(\lambda(1+q)+q\bigr)\delta$.
As explained next, our surrogate search provides $q=2+\eta$ in the continuous algorithm and $q=2$ in the discrete algorithm.
The transfer loss is $\lambda=1+\eta$ continuously and $\lambda=1$ discretely, giving the bounds $5+5\eta+\eta^2$ and $5$, respectively.

\paragraph{Macro surrogates.}
The macro search still follows the Cheng, Huang and Zhang's\cite{cheng2025constant} philosophy of looking for a short
certificate near a marked host macro.  The difference is what the certificate is
used for.  We construct an auxiliary curve $Q_M$ for each macro $M$, and a
successful search returns a subcurve $\pi\subseteq Q_M$.  This $\pi$ is fed
directly into the transfer structure; it is not converted into a subcurve of the
host block.

A macro is marked by a query curve $Z$ if some host subcurve within distance $\delta$ of $Z$ intersects with it.  
Such a witness may extend far beyond the macro and may contain many vertices.  The key point is that the extension 
has low simplification complexity at threshold $\delta$, by restriction, it matches a subcurve of $Z$, and 
when $Z$ has $O(s)$ complexity this gives an $O(s)$-vertex simplification.  At every macro boundary, 
we therefore precompute the longest vertex-aligned suffix ending there and the longest vertex-aligned prefix 
starting there whose endpoint-augmented simplifications have $O(s)$ vertices.
The \emph{auxiliary curve} of a macro concatenates the simplification of the maximal suffix at its left boundary, 
the macro itself, and the simplification of the maximal prefix at its right boundary, for a total size of $O(s)$. 
Maximality of these suffixes and prefixes forces every witness subcurve that marks the macro to lie inside the host range covered by the macro and the two maximal subcurves, 
and this range is within $(1+\eta)\delta$ of the auxiliary curve
continuously, by the approximate simplification primitive, and within
$\delta$ discretely, by the exact one, so the auxiliary curve of a marked
macro contains a surrogate of $Z$: within $(2+\eta)\delta$ continuously
and within $2\delta$ discretely.  One
free-start/free-end propagation of the auxiliary curve against $Z$
finds it in $O(s\,|Z|)$ time.  A small query is thus searched against a
small precomputed curve, never against one built at the scale of the
full block.

\paragraph{Pieces, sampling, and the sparse fallback.}
The remaining question is which macros are marked, as the algorithm
cannot afford to test them all.  Split $\sigma_l$ into \emph{pieces} of
$\mu_3$ vertices, matching a fine macro scale $s=\mu_3$, so that each
piece-against-macro search costs only $O(\mu_3^2)$.  For each piece,
sample fine macros at random and search the sampled ones.  A piece is \emph{dense} if it marks at least $\omega$ fine macros; 
dense pieces obtain a surrogate from the sample with high probability.  If every piece
obtains a surrogate, the \emph{sequential} branch chains them: it transfers the
bottom-boundary sources through the first surrogate, feeds the output
to the second, and so on to the top boundary.  Completeness is
preserved at every step by cutting the witness at piece boundaries, and
soundness does not accumulate, again because matchings concatenate at
their maximum.  If some piece fails, then on the sampling-success event
it is not dense: the host subcurves matching it can be located outright by one
$O(\mu_1\mu_3)$ time propagation, affordable precisely because the piece is
small, and they fall into fewer than $\omega$ fine macros.  Any full-block witness
restricts, by the restriction property, to a witness for the failed
piece, so it touches one of these macros.  The collected fine macros are mapped to 
their containing \emph{coarse} macros, at scale $s=\mu_2$ matching the whole
block $\sigma_l$, and each coarse auxiliary curve is tested against all
of $\sigma_l$; the first success feeds a single auxiliary-transfer query.  
The two scales are key: the sequential branch performs many cheap fine-scale searches, 
while the sparse fallback performs only a few expensive coarse-scale ones.

\paragraph{Transfer structures.}
It remains to answer the transfer queries themselves.  Answering each
by a fresh dynamic program on $\tau_k\times\pi$ would cost
$O(\mu_1\mu_3)$ per piece, and across all pieces and block pairs this
alone adds up to $\Theta(nm)$.  Instead, each auxiliary curve is
preprocessed once per host block and reused for every block pair in its
row.  The edge sequence of the auxiliary curve carries a balanced binary
decomposition; a query subcurve $\pi$ splits into logarithmically many
canonical intervals, each with a precomputed table, and the tables are
chained in query order.  Because matchings concatenate at their maximum
cost, chaining logarithmically many intervals does not degrade the
threshold, and a query costs $\widetilde O(\mu_1)$.

The continuous and discrete instantiations of this structure differ,
and the difference is where the stronger discrete bound comes from.
Continuously, the sources on a row form intervals with continuously
many positions.  The structure discretizes each source edge by portals,
precomputes the reachable top-row arrays from every portal, and merges
the retrieved arrays greedily in source order at query time; the crossing lemma \Cref{lem:crossing} shows that the greedy merge is correct.  
Rounding sources to portals is what injects the $(1+\eta)$ factor of transfer loss into the soundness
accounting.

Discretely, the free-space graph of $\tau_k$ and an auxiliary curve is
a finite planar directed acyclic graph.  The structure stores, for each
canonical interval, the farthest final-row vertex reachable from every
source and tests newly exposed targets with the constant-time planar
reachability oracle of Holm, Rotenberg, and
Thorup~\cite{holm2015planar}; the same crossing lemma shows that each final-row target 
needs to be tested only once.  No geometric rounding is needed, so the discrete transfer is exact, 
and the per-host-block preprocessing, the transfer tables together with the batched simplifications of the host block, 
is cheaper by a factor of $\mu_1$.  The cheaper preprocessing permits larger host
blocks, and this is the entire source of the gap between the two
exponents: balancing the block, piece, and sampling parameters gives
$\widetilde O(nm^{8/9})$ time continuously, up to $\eps$-dependent
factors, and
$\widetilde O(nm^{4/5})$ time discretely.

Combining the inherited block induction of~\cite{cheng2025constant} with the new
implementation of the bottom-to-top transition yields complete and
$\beta$-sound \textsc{Reach} updates.  The direct auxiliary-transfer
accounting gives a continuous fixed-threshold procedure distinguishing
$d_F\le\delta$ from $d_F>(5+\epsilon)\delta$, and the exact discrete
simplification and planar transfer primitives give a discrete fixed-threshold
procedure distinguishing $d_{dF}\le\delta$ from $d_{dF}>5\delta$.  The optimization
algorithms follow from standard decision-to-optimization conversions: the transformation
of~\cite{colombe2021approximating} continuously, and a threshold search
seeded by a coarse estimate of~\cite{bringmann2016approximability}
discretely.  The search introduces only a tunable $1+O(\eps)$ grid loss,
so the discrete optimization result is stated as a
$(5+\eps)$-approximation although its fixed-threshold decision primitive
is a $5$-gap procedure.

%% file: preliminaries.tex
%======================================================================
\section{Preliminaries} 
\label{sec:preliminaries}
%======================================================================

Throughout the paper, all curves are polygonal curves in $\R^d$, where $d$ is
a fixed constant.  The input curves are given by vertex sequences
\[
    \tau= \langle v_1,\ldots,v_n \rangle,
    \qquad
    \sigma= \langle w_1,\ldots,w_m \rangle,
\]
with $m\le n$.  We write $\dF$ and $\ddF$ for the continuous and discrete
Fr\'echet distances, respectively.

We write $\widetilde O(\cdot)$ for bounds that suppress polylogarithmic
factors in $n$ and $m$, $\widetilde O_d(\cdot)$ when the hidden factors may
also depend on the fixed dimension, and $\widetilde O_{d,\eps}(\cdot)$ when
they may additionally depend on the approximation parameter $\eps$.

In \Cref{sec:continuous,sec:discrete}, we work at a fixed threshold
$\delta>0$ and prove gap-decision procedures.  Their conversion to
approximation algorithms is described in the introduction and carried out
at the end of each of these sections.

%----------------------------------------------------------------------
\subsection{Curve notation and reachability} 
\label{subsec:curve-notation}
%----------------------------------------------------------------------

For a curve $P= \langle p_1,\ldots,p_N \rangle $, let $N = \abs{P}$.  Unless stated otherwise,
curve complexity means number of vertices; an $s$-edge subcurve has $O(s)$
vertices.  We
write $x\le_P y$ if $x$ appears no later than $y$ and write $P[x,y]$ for the
subcurve from $x$ to $y$.  For indices $i\le j$, we also write $P[i,j]$ for
the vertex subcurve $p_i,\ldots,p_j$.  When concatenating consecutive
subcurves with a common endpoint, we identify the common endpoint; keeping an
extra duplicate copy only changes complexity bounds by a constant factor.  Let
$\mathbb B(p,r)=\set{x\in\R^d:\norm{x-p}\le r}$ denote the closed Euclidean
ball of radius $r$ centered at $p$.

For a curve $Q= \langle q_1,\ldots,q_t \rangle$, its edge sequence is
$q_1q_2,\ldots,q_{t-1}q_t$.  A \emph{dyadic interval} of $Q$ means a contiguous
interval of this edge sequence.  The edge interval from $q_iq_{i+1}$ through
$q_{j-1}q_j$ represents the vertex-to-vertex subcurve $Q[i,j]$.  A continuous
subcurve $Q[x,y]$ consists of at most two partial-edge pieces and a maximal
vertex-to-vertex middle part.  A discrete singleton $Q[i,i]$ is treated as a
degenerate subcurve.

Let $P,Q$ be curves and let $r\ge0$.  For $x\le_P y$ and $a\le_Q b$, we say
that $(y,b)$ is \emph{$r$-reachable from $(x,a)$} if
\[
    \dF(P[x,y],Q[a,b])\le r.
\]
Equivalently, there is a monotone path from $(x,a)$ to $(y,b)$ in the
free-space diagram of $P$ and $Q$ at threshold $r$.  When $P=\tau$ and
$Q=\sigma$, a pair $(x,y)$ is \emph{globally $r$-reachable} if it is $r$-reachable
from $(v_1,w_1)$.

For the discrete distance, a vertex pair $(i,j)$ is feasible at threshold $r$
if $\norm{p_i-q_j}\le r$.  The discrete free-space graph contains the feasible
pairs and has directed edges from $(i,j)$ to each feasible pair among
$(i+1,j)$, $(i,j+1)$, and $(i+1,j+1)$.  For $i'\le i$ and $j'\le j$, the pair
$(i,j)$ is discretely $r$-reachable from $(i',j')$ exactly when
\[
    \ddF(P[i',i],Q[j',j])\le r.
\]
When $P=\tau$ and $Q=\sigma$, a discrete vertex pair is globally
$r$-reachable if it is discretely $r$-reachable from $(1,1)$.  The discrete
free-space graph is a planar directed acyclic graph.

We use the triangle inequality for both $\dF$ and $\ddF$.  We also use the
following two elementary consequences of monotone matchings.

\begin{lemma}[Restriction]
\label{lem:restriction}
Suppose $\dF(P,Q)\le r$.  For every subcurve $P'\subseteq P$, there is a
contiguous subcurve $Q'\subseteq Q$ such that $\dF(P',Q')\le r$.  The
analogous statement holds for $\ddF$ and contiguous vertex subcurves.
\end{lemma}

\begin{proof}
Restrict a witnessing continuous matching to the parameter interval of $P'$.
Its projection onto $Q$ is contiguous.  The discrete proof is identical.
\end{proof}

\begin{lemma}[Crossing\cite{cheng2025constant}]
\label{lem:crossing}
Suppose there are monotone free-space paths at threshold $r$ from
$(x_1,a)$ to $(y_1,b)$ and from $(x_2,a)$ to $(y_2,b)$, where
\[
    x_1\le_P x_2\le_P y_2\le_P y_1
    \qquad\text{and}\qquad
    a\le_Q b.
\]
Then $(y_2,b)$ is $r$-reachable from $(x_1,a)$.  The same statement holds in
the discrete free-space graph.
\end{lemma}

\begin{proof}
The first path starts weakly to the left of the second on row $a$ and ends
weakly to its right on row $b$.  Planarity and monotonicity force an
intersection, and splicing the two paths there gives the desired path.  The same ordering argument forces the two embedded grid paths to
intersect.  In this planar embedding, every edge is a unit horizontal
segment, a unit vertical segment, or a cell diagonal of common slope.
Distinct edges therefore meet only at shared grid vertices.  Hence the
two paths share a grid vertex and splicing at this vertex gives the desired
discrete path.
\end{proof}

%----------------------------------------------------------------------
\subsection{Curve simplification}
\label{subsec:simplification}
%----------------------------------------------------------------------

We use the following simplification primitives.  Define
\[
\begin{aligned}
    k_{\mathrm c}^*(P,\delta)
    &\defeq
    \min\set{\abs{Q}:\dF(P,Q)\le\delta},\\
    k_{\mathrm d}^*(P,\delta)
    &\defeq
    \min\set{\abs{Q}:\ddF(P,Q)\le\delta},
\end{aligned}
\]
where the comparison curve $Q$ need not be a subcurve of $P$.

\begin{lemma}[Continuous simplification~\cite{cheng2025simplification}]
\label{lem:simplification}
Let $P$ be a polygonal curve in $\R^d$ for fixed $d$.  Given $\delta>0$ and
$\eta\in(0,1)$, one can compute a polygonal curve $P'$ satisfying
\[
    \dF(P,P')\le(1+\eta)\delta,
    \qquad
    \abs{P'}\le\max\set{1,\,2k_{\mathrm c}^*(P,\delta)-2}.
\]
The running time is
\[
    O\!\left(
        \eta^{-\alpha}\abs{P}\log(1/\eta)
    \right),
    \qquad
    \alpha=2(d-1)\floor{d/2}^{\,2}+d.
\]
\end{lemma}

\begin{remark}[Endpoint augmentation]
\label{rem:endpoint-augmentation}
Whenever a simplification of a subcurve $P[x,y]$ is used as an auxiliary
curve, we add the endpoints $x$ and $y$ if they are not already present.
This endpoint augmentation increases the size by at most two and preserves
the same Fr\'echet error bound.  In the continuous setting, the added
connector at an endpoint is matched while the $P$-side stays fixed at that
endpoint; every point of the connector lies within the matched distance by
convexity.  In the discrete setting, the added endpoint vertex is followed
by a step to the old first or last vertex of the simplification while the
corresponding endpoint of $P$ stays fixed.
\end{remark}

The discrete algorithm requires simplifications of many subcurves of the
same host block, together with the exact value of $k_{\mathrm d}^*$ for
each of them.  For a single curve, an exact near-linear time algorithm is
classical~\cite{bereg2008simplifying}; \Cref{lem:batched-discrete-simplifications}
extends it to all $\Theta(N^2)$ vertex subcurves at once: after
$\widetilde{O}_d(N^2)$ preprocessing, each exact value
$k_{\mathrm d}^*(P[i,j],\delta)$ is an $O(1)$-time lookup, and a
simplification of exactly this size at error exactly $\delta$ is
implicitly represented.  The proof rests on a characterization of
$k_{\mathrm d}^*$ by ball partitions (\Cref{lem:ball-partition}), which
we also use on its own later, and a greedy optimality statement
(\Cref{lem:greedy-domination}).

Call a set of consecutive vertices of $P$ a \emph{$\delta$-block} if it is
contained in some ball of radius $\delta$.  Every subblock of a
$\delta$-block is a $\delta$-block.

\begin{lemma}[Ball partition characterization]
\label{lem:ball-partition}
Let $P=\langle p_1,\ldots,p_N\rangle$ be a point sequence in $\R^d$, and let
$\delta>0$.  Fix $1\le i\le j\le N$ and $k\ge 1$.  Then
$k_{\mathrm d}^*(P[i,j],\delta)\le k$ holds if and only if
$p_i,\ldots,p_j$ can be partitioned into at most $k$ contiguous nonempty
$\delta$-blocks.  Moreover, suppose that $\ddF(P[i,j],Z)\le\delta$ for a
point sequence $Z=(z_1,\ldots,z_k)$.  Then such a partition into at most $k$
blocks exists in which each block is contained in $\mathbb B(z_s,\delta)$ for some
$s$, distinct blocks use distinct indices $s$, and these indices increase
along the partition.
\end{lemma}

\begin{proof}
($\Leftarrow$)  Let $B_1,\ldots,B_t$ with $t\le k$ be the blocks in order.
For each $r$, pick a center $c_r$ with $B_r\subseteq \mathbb B(c_r,\delta)$.  Pair
every index of $B_r$ with $c_r$.  Inside $B_r$, only the index on the $P$
side advances.  From the last index of $B_r$, both indices advance
simultaneously to the first index of $B_{r+1}$ paired with $c_{r+1}$.  The
result is a monotone matching between $P[i,j]$ and $(c_1,\ldots,c_t)$ in
which every pair is within distance $\delta$.  Hence
$k_{\mathrm d}^*(P[i,j],\delta)\le t\le k$.

($\Rightarrow$, together with the moreover part)  Let $C$ be a monotone
matching that certifies $\ddF(P[i,j],Z)\le\delta$.  For $s\in[k]$, let $I_s$
be the set of indices matched with $z_s$.  Along $C$, both indices are nondecreasing, 
and each step advances each index by at most one.  The $Z$-side index starts at $1$ and ends at $k$
and never skips a value, so $I_s$ is nonempty.  Moreover, the pairs with
$Z$-side index $s$ form a contiguous run of $C$, along which the $P$-side
index advances by at most one per step, so $I_s$ consists of consecutive
integers.  Every $a\in I_s$ satisfies $\norm{p_a-z_s}\le\delta$.  Set $m_0:=i-1$ and
$m_s:=\max I_s$ for $s\in[k]$.  We claim that $\min I_s\le m_{s-1}+1$ for
every $s$.  For $s=1$, the claim holds because $C$ starts at the pair
$(i,1)$.  For $s\ge 2$, consider the step of $C$ that leaves the pair
$(m_{s-1},s-1)$.  By the maximality of $m_{s-1}$, this step advances the
index on the $Z$ side.  Hence the next pair is $(m_{s-1},s)$ or
$(m_{s-1}+1,s)$, and the claim follows.  The claim gives
$[m_{s-1}+1,m_s]\subseteq I_s$ whenever $m_s>m_{s-1}$.  The nonempty ranges
$[m_{s-1}+1,m_s]$ for $s=1,\ldots,k$ are contiguous and appear in order.
They jointly cover $[i,j]$ because $m_k=j$.  The range for index $s$ is
contained in $\mathbb B(z_s,\delta)$.  Deleting the empty ranges yields the asserted
partition into at most $k$ $\delta$-blocks.
\end{proof}

\begin{corollary}
\label{cor:kstar-witness}
If $\ddF(A,Z)\le\delta$ for point sequences $A$ and $Z$, then
$k_{\mathrm d}^*(A,\delta)\le\abs{Z}$.
\end{corollary}

\begin{proof}
Apply the moreover part of \Cref{lem:ball-partition} to $A$ and $Z$, and
then apply the equivalence with $k=\abs{Z}$.
\end{proof}

\begin{lemma}[Greedy optimality]
\label{lem:greedy-domination}
Fix $1\le i\le j\le N$.  The greedy partition of $P[i,j]$ into
$\delta$-blocks is defined as follows: repeatedly remove the longest prefix
of the remaining sequence that is a $\delta$-block.  Let
$i-1=s_0<s_1<\cdots<s_q=j$ be its boundaries.  Let
$i-1=e_0<e_1<\cdots<e_p=j$ be the boundaries of an arbitrary partition of
$P[i,j]$ into contiguous nonempty $\delta$-blocks.  Then $e_r\le s_r$ for
all $r\le\min\{p,q\}$.  Consequently, $p\ge q$ and
\[
    q \;=\; k_{\mathrm d}^*(P[i,j],\delta).
\]
\end{lemma}

\begin{proof}
We prove $e_r\le s_r$ by induction on $r$.  If $s_r=j$, the claim is
trivial.  So assume $s_r<j$.  By the maximality of the $r$-th greedy block,
$P[s_{r-1}+1,s_r+1]$ is not a $\delta$-block.  Consider first $r=1$, and
suppose that $e_1>s_1$.  Then $P[i,s_1+1]$ is a subblock of the
$\delta$-block $P[i,e_1]$, so it is itself a $\delta$-block.  This is a
contradiction.  Now let $r>1$, assume $e_{r-1}\le s_{r-1}$, and suppose that
$e_r>s_r$.  Then
\[
    P[s_{r-1}+1,\,s_r+1]\;\subseteq\;P[e_{r-1}+1,\,e_r].
\]
The right side is the $r$-th block of the given partition, so the left side
is a $\delta$-block.  This is again a contradiction, and the domination
follows.  If $p<q$, then $j=e_p\le s_p<s_q=j$, which is absurd.  Hence every
partition of $P[i,j]$ into contiguous nonempty $\delta$-blocks has at least
$q$ blocks.  The greedy partition is itself such a partition, so $q$ is the
minimum number of blocks.  By \Cref{lem:ball-partition}, this minimum equals
$k_{\mathrm d}^*(P[i,j],\delta)$.
\end{proof}

\begin{lemma}[Batched discrete simplifications]
\label{lem:batched-discrete-simplifications}
Let $P=(p_1,\ldots,p_N)$ be a point sequence in $\R^d$ for fixed
$d$, and let $\delta>0$.  One can preprocess $P$ deterministically in $\widetilde O_d(N^2)$ time and
$O(N^2)$ space so that the following holds for every $1\le i\le j\le N$.
\begin{enumerate}
    \item The exact value $k_{\mathrm d}^*(P[i,j],\delta)$ is returned in
    $O(1)$ time.
    \item The data structure implicitly represents a simplification
    $\zeta_{i,j}$ with
    \[
        \ddF(P[i,j],\zeta_{i,j})\le\delta,
        \qquad
        \abs{\zeta_{i,j}} \;=\; k_{\mathrm d}^*(P[i,j],\delta).
    \]
    The curve $\zeta_{i,j}$ is represented as a prefix of a stored center
    sequence, together with a compact monotone block matching.  Both can be
    materialized in $O(\abs{\zeta_{i,j}})$ time.
\end{enumerate}
\end{lemma}

\begin{proof}
The subcurve $P[b,e]$ is a $\delta$-block if and only if its smallest
enclosing ball has radius at most $\delta$.  Computing the smallest
enclosing ball is an LP-type problem of combinatorial dimension $d+1$.  For
fixed $d$, its center and radius can be computed in time linear in the
number of points, with a multiplicative constant that depends only on $d$: deterministically by \cite{chazelle1996linear}, or by the simpler
randomized algorithm of \cite{welzl2005smallest} in expected linear time.  With the
randomized algorithm, the preprocessing bound below holds in expectation.
The resulting feasibility test is exact.  It is also monotone: if $P[b,e]$
is a $\delta$-block, then so is $P[b,e']$ for every $e'\in[b,e]$.

\emph{Preprocessing.}
Fix a starting index $i$.  We compute the boundaries
$i-1=s_0<s_1<\cdots<s_q=N$ of the greedy partition of $P[i,N]$.  Each
boundary is found by exponential search followed by binary search.
Starting from $b=s_{r-1}+1$, we test the prefixes ending at
$b,\,b+1,\,b+3,\,b+7,\ldots$ until the first infeasible endpoint or the end
of the curve.  We then binary search for the maximal feasible endpoint in
the last gap.  This takes $O(\log N)$ feasibility tests, and each test runs
on a window of length at most twice the resulting block.  Hence the $r$-th
block costs $O(L_r\log N)$ time, where $L_r=s_r-s_{r-1}$, and the whole run
from $i$ costs $O((N-i+1)\log N)$ time.  For this start we store three
items.  The first item is the boundary sequence.  The second item is the
sequence of centers $y_1,\ldots,y_q$, where $y_r$ is the center of the
exact smallest enclosing ball of the greedy block $P[s_{r-1}+1,s_r]$; each
center takes one additional computation is linear in the length of the block.  The third item is
the table
\[
    t(i,j)\;=\;\min\{\,t : s_t\ge j\,\}
    \qquad\text{for all } j\ge i,
\]
filled by one scan of the boundaries.  Summing over all starting indices $i$ gives $O(N^2\log N)$ time, where the
hidden constant depends only on $d$, that is, $\widetilde O_d(N^2)$ time,
and $O(N^2)$ space in total.

\emph{Query (1).}
Write $t=t(i,j)$.  We claim that the greedy partition of $P[i,j]$ has the
boundaries $\min\{s_r,j\}$ for $r=1,\ldots,t$.  For $r<t$, by induction on $r$, the $r$-th block
of $P[i,j]$ coincides with the $r$-th block of the run from $i$.  
This block is a $\delta$-block, and its maximality certificate is
inherited: $P[s_{r-1}+1,s_r+1]$ is not a $\delta$-block, and
$s_r+1\le s_{t-1}+1\le j$, so the certificate lies inside $[i,j]$.  The
final block $P[s_{t-1}+1,j]$ is a subblock of the greedy block
$P[s_{t-1}+1,s_t]$, so it is a $\delta$-block, and it ends at $j$.  In
other words, taking longest $\delta$-block prefixes commutes with
truncation at $j$.  Applying \Cref{lem:greedy-domination} to $P[i,j]$ gives
\[
    k_{\mathrm d}^*(P[i,j],\delta)\;=\;t(i,j).
\]
The query is one array lookup.

\emph{Query (2).}
Set $\zeta_{i,j}:=(y_1,\ldots,y_t)$ with $t=t(i,j)$.  Its block matching
pairs
\[
    P[s_{r-1}+1,\,\min\{s_r,j\}]
    \quad\text{with } y_r,
    \qquad r=1,\ldots,t.
\]
Every full greedy block lies in $\mathbb B(y_r,\delta)$, because the block is a
$\delta$-block and $y_r$ is the center of its smallest enclosing ball.  The
truncated final block is a subset of the $t$-th full block, so the same
center serves.  By the construction in the ($\Leftarrow$) direction of
\Cref{lem:ball-partition}, $\ddF(P[i,j],\zeta_{i,j})\le\delta$.  By
Query~(1), $\abs{\zeta_{i,j}}=t(i,j)=k_{\mathrm d}^*(P[i,j],\delta)$.  The
prefix and the matching intervals are read off the stored arrays in
$O(\abs{\zeta_{i,j}})$ time.
\end{proof}

%----------------------------------------------------------------------
\subsection{Reachability primitives}
\label{subsec:local-propagation}
%----------------------------------------------------------------------

We repeatedly apply the standard free-space dynamic program on small
subproblems.  Given curves $P,Q$, a threshold $r$, and initial reachable
portions on the left and bottom boundaries of $P\times Q$, standard local
free-space propagation computes the reachable portions on the right and top
boundaries in $O(\abs{P}\abs{Q})$ time~\cite{alt1995computing}.  The discrete
analogue is the usual propagation in the monotone grid graph and has the same
running time.

We also use a free-start/free-end variant.  Initialize every free point on the
bottom boundary $P\times\{q_1\}$ as a source and run the same propagation.
A point $(y,q_{\abs Q})$ is reachable if and only if there exists
$x\le_P y$ such that
\[
    \dF(P[x,y],Q)\le r.
\]
Thus one computation finds a subcurve of $P$ matching all of $Q$, if one
exists, in $O(\abs{P}\abs{Q})$ time.  The discrete analogue initializes all
feasible vertices on the first row.

The discrete algorithm additionally uses the following oracle.

\begin{theorem}[Planar reachability oracle~\cite{holm2015planar}]
\label{thm:planar-reachability}
A planar directed graph with $N$ vertices can be preprocessed in $O(N)$ time
and space so that reachability between any two vertices can be tested in
$O(1)$ time.
\end{theorem}

These primitives are applied only to local curves or preprocessed auxiliary
free-space graphs; the algorithms never run the quadratic dynamic programming algorithm on
the full input diagram.

\paragraph{Certified simplifications.}
Whenever a later routine needs an explicit matching between a curve $P$ and a
simplification $\zeta$, we recover one using the corresponding local dynamic programming algorithm. 
 In the continuous setting, if $\dF(P,\zeta)\le r$, standard
free-space propagation with predecessor information recovers a monotone
free-space path $\Gamma$ in $O(\abs{P}\abs{\zeta})$ time.  In the discrete
setting, if $\ddF(P,\zeta)\le r$, backtracking the discrete dynamic programming algorithm
recovers a monotone grid path $\Gamma$ within the same time bound.  In either
case, we store $\Gamma$ in compressed form of size $O(\abs{P}+\abs{\zeta})$.

Images and preimages are taken with respect to this stored matching path.  In
the continuous setting, write
\[
    \Gamma(t)=(\Gamma_P(t),\Gamma_\zeta(t)),
\]
where both coordinates are nondecreasing.  For an interval $X\subseteq P$ and
an interval $Y\subseteq \zeta$, define
\[
    M_{P,\zeta}(X)
    :=
    \Gamma_\zeta\bigl(\Gamma_P^{-1}(X)\bigr),
    \qquad
    M^{-1}_{P,\zeta}(Y)
    :=
    \Gamma_P\bigl(\Gamma_\zeta^{-1}(Y)\bigr).
\]
The discrete definitions are identical, with $\Gamma$ viewed as a monotone
grid path and $X,Y$ as contiguous vertex intervals.

By monotonicity of $\Gamma$, the image and preimage of an interval are again
intervals, and they can be computed by a linear scan of the stored path.  For
a single point $x\in P$, the set $M_{P,\zeta}(x)$ may be an interval of
$\zeta$; when a single representative is needed, we choose any
$\bar x\in M_{P,\zeta}(x)$.

If $x\le_P y$, $\bar x\in M_{P,\zeta}(x)$, and
$\bar y\in M_{P,\zeta}(y)$ occur in this order along $\Gamma$, then the
restricted path certifies
\[
    \dF\bigl(P[x,y],\zeta[\bar x,\bar y]\bigr)\le r.
\]
The analogous statement holds in the discrete setting: if $i\le j$ and
$\bar i\in M_{P,\zeta}(i)$, $\bar j\in M_{P,\zeta}(j)$ occur in order along
the stored grid path, then
\[
    \ddF\bigl(P[i,j],\zeta[\bar i,\bar j]\bigr)\le r.
\]

%----------------------------------------------------------------------
\subsection{Blocks and approximate reachable sets}
\label{subsec:block-framework}
%----------------------------------------------------------------------

Let $\mu_1$ and $\mu_2$ be block-size parameters.  Partition $\tau$ into
consecutive blocks of complexity $\Theta(\mu_1)$ and $\sigma$ into consecutive
blocks of complexity $\Theta(\mu_2)$.  Standard rounding and padding affect
only constant factors.  Write
\[
    \tau_k=\tau[a_k,a_{k+1}],
    \qquad
    \sigma_l=\sigma[b_l,b_{l+1}]
\]
for one pair of blocks, where consecutive blocks share their boundary point.

The four boundaries of $\tau_k\times\sigma_l$ are
\[
\begin{aligned}
    L_{k,l}&=\set{a_k}\times\sigma_l,
    &\qquad
    R_{k,l}&=\set{a_{k+1}}\times\sigma_l,\\
    B_{k,l}&=\tau_k\times\set{b_l},
    &
    T_{k,l}&=\tau_k\times\set{b_{l+1}}.
\end{aligned}
\]
The left and bottom boundaries are incoming, while the right and top
boundaries are outgoing.  In the discrete setting, the same notation denotes
the corresponding boundary vertices of the grid block.

Fix $\delta>0$ and $\beta\ge1$.  A set on a block boundary is
\emph{complete} if it contains every globally $\delta$-reachable point on that
boundary, and is \emph{$\beta$-sound} if every point it contains is globally
$\beta\delta$-reachable.  When such a set is used as input to a local
propagation at threshold $\delta$, we denote by $A^\delta$ its intersection
with the corresponding $\delta$-free boundary portion.  This clipping preserves
all sources that can participate in a local $\delta$-witness through the
block, and preserves $\beta$-soundness because $A^\delta\subseteq A$.

We use without further mention that taking unions of complete and
$\beta$-sound boundary sets preserves both properties.

In the continuous algorithm, boundary sets are represented by intervals on
boundary edges.  In the discrete algorithm, they are represented by subsets
of boundary vertices.  We write \REACH{} for a local block update that receives
sets on $L_{k,l}\cup B_{k,l}$ and returns sets on
$R_{k,l}\cup T_{k,l}$.  Each implementation of \REACH{} will map complete and
$\beta$-sound incoming sets to complete and $\beta$-sound outgoing sets.

%% file: continuous.tex
%======================================================================
\section{Continuous Fr\'echet Algorithm}
\label{sec:continuous}
%======================================================================

We prove the continuous result by a fixed-threshold gap-decision procedure.
The decision-to-approximation conversion is given at the end of the section.

\begin{theorem}[Continuous gap decision]
\label{thm:continuous-gap}
Let $\tau$ and $\sigma$ be polygonal curves in fixed dimension, with
$\abs{\tau}=n$ and $\abs{\sigma}=m\le n$.  For any $\epsilon>0$ and threshold
$\delta>0$, there is a randomized gap-decision procedure that accepts if
$\dF(\tau,\sigma)\le\delta$ and rejects if
$\dF(\tau,\sigma)>(5+\epsilon)\delta$, with high probability, in
$\widetilde O_{d,\epsilon}(n m^{8/9})$ time.
\end{theorem}

It suffices to consider $\eps\in(0,1/2]$, for larger $\eps$, run the
procedure with accuracy parameter $1/2$.  Throughout the section we fix
$\delta>0$ and set the internal accuracy parameter $\eta\defeq\eps/6$,
so that
\[
    (1+\eta)(3+\eta)\delta+(2+\eta)\delta
    =
    (5+5\eta+\eta^2)\delta
    \le
    (5+\eps)\delta,
\]
and we set $\beta\defeq 5+\eps$.

%----------------------------------------------------------------------
\subsection{Fixed-threshold block dynamic program}
\label{subsec:continuous-block-dp}
%----------------------------------------------------------------------

The algorithm processes the block pairs $\tau_k\times\sigma_l$ in any
topological order of the block grid from \Cref{subsec:block-framework}.  When
a block pair is processed, the reachable interval arrays stored on its
incoming boundaries $L_{k,l}\cup B_{k,l}$ are passed to the local update
\REACH{}, and the returned arrays are stored on $R_{k,l}\cup T_{k,l}$.
Shared boundaries are represented once and are updated by taking unions.

On the two outer incoming boundaries $\{v_1\}\times\sigma$ and
$\tau\times\{w_1\}$, the algorithm stores the exact portions reachable from
$(v_1,w_1)$ at threshold $\delta$, computed by one-dimensional free-space
propagation.  Every other boundary set is produced by a block update.  After
all blocks have been processed, the algorithm accepts exactly when
$(v_n,w_m)$ is stored.

\begin{lemma}[Block induction]
\label{lem:continuous-block-induction}
Assume that each local \REACH{} update maps complete and $\beta$-sound
incoming sets to complete and $\beta$-sound outgoing sets.  Then the block
dynamic program accepts whenever $\dF(\tau,\sigma)\le\delta$ and rejects
whenever $\dF(\tau,\sigma)>\beta\delta$.
\end{lemma}

\begin{proof}
We induct over the block order.  The invariant is that every stored boundary
set is complete and $\beta$-sound.  It holds initially because the two outer
incoming boundaries are computed exactly at threshold $\delta$.

Consider a block $\tau_k\times\sigma_l$ when it is processed.  Its incoming
sets have already been initialized or produced, hence satisfy the invariant.
By assumption, the \REACH{} call produces complete and $\beta$-sound outgoing
sets.  Taking unions on shared boundaries preserves completeness and
$\beta$-soundness.

If $\dF(\tau,\sigma)\le\delta$, then the final pair $(v_n,w_m)$ is globally
$\delta$-reachable and is stored by completeness.  Conversely, if the
algorithm accepts, then $\beta$-soundness of the final stored set gives a
monotone Fr\'echet matching of cost at most $\beta\delta$ from
$(v_1,w_1)$ to $(v_n,w_m)$.  The local certificates concatenate in block
order, and the cost of the concatenated matching is the maximum cost of its
pieces, not their sum.
\end{proof}

%----------------------------------------------------------------------
\subsection{Local boundary propagation}
\label{subsec:continuous-local-boundary}
%----------------------------------------------------------------------

Fix a block pair $\tau_k\times\sigma_l$, and let
$A_L\subseteq L_{k,l}$ and $A_B\subseteq B_{k,l}$ be the incoming arrays.  A
\emph{local $\delta$-witness} for a point on $R_{k,l}\cup T_{k,l}$ is a
monotone free-space path at threshold $\delta$, contained in the block, from
a globally $\delta$-reachable incoming point to that outgoing point.

The four witness classes are determined by their incoming and outgoing
boundaries:
\[
\begin{array}{c|c|c}
\text{class} & \text{source and target} & \text{curve condition}\\ \hline
\text{left-to-right local propagation}
& (a_k,p)\to(a_{k+1},q)
& \dF(\tau_k,\sigma_l[p,q])\le\delta
\\[1mm]
\text{bottom-to-right suffix propagation}
& (x,b_l)\to(a_{k+1},q)
& \dF(\tau_k[x,a_{k+1}],\sigma_l[b_l,q])\le\delta
\\[1mm]
\text{left-to-top prefix propagation}
& (a_k,p)\to(y,b_{l+1})
& \dF(\tau_k[a_k,y],\sigma_l[p,b_{l+1}])\le\delta
\\[1mm]
\text{auxiliary-transfer propagation}
& (x,b_l)\to(y,b_{l+1})
& \dF(\tau_k[x,y],\sigma_l)\le\delta .
\end{array}
\]
In each row, the source is required to be globally $\delta$-reachable.  The
first two classes contribute to the right boundary, and the last two
contribute to the top boundary.

\begin{lemma}[Witness classification]
\label{lem:continuous-witness-classification}
Every globally $\delta$-reachable point of $R_{k,l}\cup T_{k,l}$ has a local
witness in one of the four classes above.
\end{lemma}

\begin{proof}
Let $z\in R_{k,l}\cup T_{k,l}$ be globally $\delta$-reachable, and let
$\Gamma$ be a monotone free-space path from $(v_1,w_1)$ to $z$.  Let $s$ be
the first point of the maximal suffix of $\Gamma$ contained in the closed
block $\tau_k\times\sigma_l$.  Monotonicity implies
$s\in L_{k,l}\cup B_{k,l}$; entry through an outgoing side can occur only at
a corner that also belongs to an incoming side. The subpath from $s$ to $z$ is a local witness.  
Its incoming and outgoing sides place it in one of the classes in the table.  At the shared corner
of $R_{k,l}$ and $T_{k,l}$, the right-exiting and top-exiting curve
conditions coincide.  A witness ending at this corner therefore belongs to
a class on each side.
\end{proof}

We now construct the contributions $I_{\mathrm{LR}}$, $I_{\mathrm{BR}}$, and
$I_{\mathrm{LT}}$ for the first three classes.  Identify $A_L$ and $A_B$
with their projections onto $\sigma_l$ and $\tau_k$, respectively, and clip
them to the $\delta$-free portions
$A_L^\delta\defeq\set{p\in A_L:\norm{a_k-p}\le\delta}$ and
$A_B^\delta\defeq\set{x\in A_B:\norm{x-b_l}\le\delta}$.  This clipping
removes no source of a true $\delta$-witness.

Whenever a simplification $\zeta$ of a subcurve $P$ is used, its endpoints
are added if necessary; endpoint augmentation increases the size by at most
two and preserves the Fr\'echet error bound by convexity.  We recover a
certified monotone matching path $M_P$ as in
\Cref{subsec:local-propagation}.  For a set $X\subseteq P$, let $M_P(X)$ be
its projection onto $\zeta$ along this path, and define $M_P^{-1}$
symmetrically.  If $(x,\bar x)$ and $(y,\bar y)$ occur in this order on
$M_P$, then $\dF\bigl(P[x,y],\zeta[\bar x,\bar y]\bigr)\le(1+\eta)\delta$.
The projection of an interval is an interval, so all images and preimages
used below are ordered interval lists of total linear complexity and are
obtained by one scan of the stored path.

For each host block $\tau_k$, retain the following simplifications, using a
sufficiently large constant hidden in the $O(\mu_2)$ bounds, an
endpoint-augmented simplification $\zeta_k^{\mathrm{all}}$ of all of
$\tau_k$, provided it has $O(\mu_2)$ vertices; the simplification
$\zeta_k^{\mathrm{suf}}$ of the longest vertex-aligned suffix
$P_k^{\mathrm{suf}}=\tau_k[s_k,a_{k+1}]$ whose endpoint-augmented
simplification has $O(\mu_2)$ vertices; and the simplification
$\zeta_k^{\mathrm{pre}}$ of the analogous longest vertex-aligned prefix
$P_k^{\mathrm{pre}}=\tau_k[a_k,t_k]$.  The preprocessing cost is included in
\Cref{subsec:continuous-completion}.  Set
$r_{\mathrm{loc}}\defeq(2+\eta)\delta$.

For left-to-right local propagation, run standard local free-space
propagation on $\zeta_k^{\mathrm{all}}\times\sigma_l$ at threshold
$r_{\mathrm{loc}}$, initialized by $A_L^\delta$ on the left boundary, and
let $W_{\mathrm{LR}}\subseteq\sigma_l$ be the reached right-boundary
portion; if the simplification was not retained, the output is empty.  For
bottom-to-right suffix propagation, define
$S_{\mathrm{suf}}\defeq M_k^{\mathrm{suf}}\bigl(A_B^\delta\cap
P_k^{\mathrm{suf}}\bigr)$, run local propagation on
$\zeta_k^{\mathrm{suf}}\times\sigma_l$ at threshold $r_{\mathrm{loc}}$,
initialized by $S_{\mathrm{suf}}$ on the bottom boundary, and let
$W_{\mathrm{BR}}\subseteq\sigma_l$ be the reached right-boundary portion.
Every source is feasible, if $x\in A_B^\delta$ and
$\bar x\in M_k^{\mathrm{suf}}(x)$, then
$\norm{\bar x-b_l}\le\norm{\bar x-x}+\norm{x-b_l}\le(2+\eta)\delta$.
For left-to-top prefix propagation, run local propagation on
$\zeta_k^{\mathrm{pre}}\times\sigma_l$ at threshold $r_{\mathrm{loc}}$,
initialized by $A_L^\delta$ on the left boundary, and let
$W_{\mathrm{LT}}\subseteq\zeta_k^{\mathrm{pre}}$ be the reached top-boundary
portion.  Define
\[
\begin{aligned}
    I_{\mathrm{LR}}
    &\defeq
    \set{(a_{k+1},q)\in R_{k,l}:
          q\in W_{\mathrm{LR}},\ \norm{a_{k+1}-q}\le\delta},\\
    I_{\mathrm{BR}}
    &\defeq
    \set{(a_{k+1},q)\in R_{k,l}:
          q\in W_{\mathrm{BR}},\ \norm{a_{k+1}-q}\le\delta},\\
    I_{\mathrm{LT}}
    &\defeq
    \set{(y,b_{l+1})\in T_{k,l}:
          y\in(M_k^{\mathrm{pre}})^{-1}(W_{\mathrm{LT}}),\
          \norm{y-b_{l+1}}\le\delta}.
\end{aligned}
\]

\begin{lemma}[Local boundary propagation]
\label{lem:local-boundary-propagation}
Assume that $A_L$ and $A_B$ are complete and $\beta$-sound.  Then
$I_{\mathrm{LR}}\cup I_{\mathrm{BR}}\cup I_{\mathrm{LT}}$ contains every
outgoing point with a left-to-right, bottom-to-right, or left-to-top
$\delta$-witness.  Every point in this union is globally
$\beta\delta$-reachable.  After host-block preprocessing, the construction
takes $\widetilde O(\mu_1+\mu_2^2)$ time per block pair.
\end{lemma}

\begin{proof}
For a left-to-right witness from $(a_k,p)$ to $(a_{k+1},q)$, completeness
gives $p\in A_L^\delta$.  Moreover,
$\dF(\tau_k,\sigma_l[p,q])\le\delta$.  Since
$\sigma_l[p,q]$ has $O(\mu_2)$ vertices, the retained full-block
simplification exists, and the triangle inequality gives
$\dF(\zeta_k^{\mathrm{all}},\sigma_l[p,q])\le(2+\eta)\delta$, so the
propagation reaches $q$ and $(a_{k+1},q)\in I_{\mathrm{LR}}$.

For a bottom-to-right witness starting at $(x,b_l)$, we have
$\norm{x-b_l}\le\delta$ and
$\dF(\tau_k[x,a_{k+1}],\sigma_l[b_l,q])\le\delta$.  If $x$ is a vertex, set $x^-=x$; otherwise let
$x^-$ be the preceding host vertex on the edge containing $x$.  The curve
$R_x\defeq\tau_k[x^-,x]\circ\overline{xb_l}\circ\sigma_l[b_l,q]$ has
$O(\mu_2)$ vertices and satisfies $\dF(\tau_k[x^-,a_{k+1}],R_x)\le\delta$:
match $\tau_k[x^-,x]$ identically, traverse the connector
$\overline{xb_l}$---which lies inside $\mathbb B(x,\delta)$ by convexity---while the
host side stays at $x$, and then use the witness matching between
$\tau_k[x,a_{k+1}]$ and $\sigma_l[b_l,q]$.  Hence the vertex-aligned suffix
$\tau_k[x^-,a_{k+1}]$ is admissible for the definition of
$P_k^{\mathrm{suf}}$, and maximality gives $s_k\le_T x^-\le_T x$, so
$x\in P_k^{\mathrm{suf}}$.  Choose any $\bar x\in M_k^{\mathrm{suf}}(x)$.
The restricted simplification matching and the witness imply
\[
    \dF\bigl(
        \zeta_k^{\mathrm{suf}}
        [\bar x,\operatorname{last}(\zeta_k^{\mathrm{suf}})],
        \sigma_l[b_l,q]
    \bigr)
    \le (2+\eta)\delta,
\]
so the suffix propagation reaches $q$ and
$(a_{k+1},q)\in I_{\mathrm{BR}}$.

For a left-to-top witness ending at $(y,b_{l+1})$, we have
$\norm{y-b_{l+1}}\le\delta$ and
$\dF(\tau_k[a_k,y],\sigma_l[p,b_{l+1}])\le\delta$.  If $y$ is a vertex, set $y^+=y$; otherwise let
$y^+$ be the next host vertex on the edge containing $y$.  The curve
$R_y\defeq\sigma_l[p,b_{l+1}]\circ\overline{b_{l+1}y}\circ\tau_k[y,y^+]$ has
$O(\mu_2)$ vertices and satisfies $\dF(\tau_k[a_k,y^+],R_y)\le\delta$, match
$\tau_k[a_k,y]$ to $\sigma_l[p,b_{l+1}]$, traverse the connector while the
host side stays at $y$, and then match $\tau_k[y,y^+]$ identically.
Therefore the vertex-aligned prefix $\tau_k[a_k,y^+]$ is admissible for the
definition of $P_k^{\mathrm{pre}}$, so maximality gives
$y\le_T y^+\le_T t_k$.  The restricted simplification matching places some
$\bar y\in M_k^{\mathrm{pre}}(y)$ in $W_{\mathrm{LT}}$, hence
$(y,b_{l+1})\in I_{\mathrm{LT}}$.

For soundness, every inserted point has a certificate at threshold
$(2+\eta)\delta$ in a simplified diagram.  Composing it with the
corresponding simplification matching gives a local certificate in the
original block of cost at most
$(3+2\eta)\delta\le(5+\eps)\delta=\beta\delta$.  Its source belongs to
$A_L$ or $A_B$ and is globally $\beta\delta$-reachable, so concatenating the
certificates proves global $\beta\delta$-reachability.

Each of the three local propagations has complexity $O(\mu_2^2)$.  Images
and preimages under the stored matchings are computed by linear scans in
$\widetilde O(\mu_1)$ total time.
\end{proof}

\subsection{Dyadic auxiliary transfer}
\label{subsec:continuous-dyadic-auxtrans}

Fix a host block $T=\tau_k$ and an auxiliary curve
$Q=(q_1,\ldots,q_t)$.  Let $e_1,\ldots,e_{N_T}$ be the host edges of $T$ in
curve order, where $N_T=O(\mu_1)$.  Internal host vertices are assigned to
their later incident edge, and the two endpoints of $T$ are assigned to
their unique incident edge.

An \emph{interval array} stores one possibly empty interval on each host
edge, and we identify the array with the subset of $T$ that it represents.
On a fixed row of a continuous free-space diagram, the reachable portion of
a host edge is a suffix of its free interval, after reaching one point, a
matching may advance along the host edge while the other curve stays fixed.
All transfer sets below are represented by interval arrays.  For a subcurve
$\pi\subseteq Q$, a threshold $\rho>0$, and an interval array $S$ on $T$,
define the ideal transfer set
\[
    \mathcal T_Q^*(\pi,\rho,S)
    \defeq
    \set{
        y\in T:
        \exists x\in S,\ x\le_T y,\
        \dF(T[x,y],\pi)\le \rho
    }.
\]
All transfer tables below are built for a fixed threshold $\rho$; in
\Cref{subsec:continuous-auxiliary-propagation}, this threshold is always
$(3+\eta)\delta$.

\begin{lemma}[Fixed-interval auxiliary transfer]
\label{lem:interval-transfer}
Fix a threshold $\rho>0$, and let $J=Q[q_a,q_b]$ be a vertex-to-vertex
interval of $Q$.  One can build a data structure for this interval and
threshold with preprocessing time $O(\mu_1^2|J|/\eta)$ and query time
$O(\mu_1)$.  Given an interval array $S$ on $T$, let
$\mathsf{Transfer}_J(\rho,S)$ be the returned interval array.  Then
\[
    \mathcal T_Q^*(J,\rho,S)
    \subseteq
    \mathsf{Transfer}_J(\rho,S)
    \subseteq
    \mathcal T_Q^*(J,(1+\eta)\rho,S).
\]
\end{lemma}

\begin{proof}
For every host edge $e_i$, discretize the free interval
$e_i\cap\mathbb B(q_a,\rho)$ by portals at spacing at most $\eta\rho$,
including the left endpoint whenever the interval is nonempty.  The interval
length is at most $2\rho$, so each host edge receives $O(1/\eta)$ portals,
and every point $z\in e_i\cap\mathbb B(q_a,\rho)$ has a portal $p\le_T z$ on
the same edge with $\norm{p-z}\le\eta\rho$.  For every source edge $e_i$ and
portal $p$, run the standard free-space DP on $T\times J$ at threshold
$\rho$, initialized from the single source $(p,q_a)$, and store the
reachable interval array on the top row $q_b$; call it $R(i,p)$.  The array is indexed by host edge and stored together with its farthest
nonempty point.  The entry of any host edge is therefore accessed in
$O(1)$ time.
One DP run costs $O(\mu_1|J|)$ time, and there are $O(\mu_1/\eta)$ source
edge/portal pairs, giving $O(\mu_1^2|J|/\eta)$ preprocessing time and the
same asymptotic space.

Given a query array $S$, first clip the source on each edge to the points
that are free with $q_a$, setting $L_i\defeq S[e_i]\cap\mathbb B(q_a,\rho)$.
If $L_i=\emptyset$, edge $e_i$ contributes nothing.  Otherwise let $\ell_i$
be the leftmost point of $L_i$, choose the rightmost portal
$p_i\le_T\ell_i$, and retrieve $R(i,p_i)$.  The starting edge is clipped at
the true leftmost feasible source:
\[
    \widetilde R_i[e_h]=
    \begin{cases}
        \emptyset, & h<i,\\
        R(i,p_i)[e_i]\cap[\ell_i,\operatorname{end}(e_i)], & h=i,\\
        R(i,p_i)[e_h], & h>i.
    \end{cases}
\]
The output is the union of the nonempty arrays $\widetilde R_i$, computed by
a left-to-right greedy merge, maintain the current farthest covered point
$f$; when processing $\widetilde R_i$ with farthest point $f_i$, discard the
array if $f_i\le_T f$, and otherwise append only the portion of
$\widetilde R_i$ strictly after $f$ and update $f=f_i$.

The merge is correct by the crossing lemma.  Suppose an earlier array
$\widetilde R_j$ advanced the frontier to $f$, and let
$w\in\widetilde R_i$ with $w\le_T f$ and $j<i$.  The portals satisfy
$p_j\le_T p_i\le_T w\le_T f$, and there is a threshold-$\rho$ path from
$(p_i,q_a)$ to $(w,q_b)$ and one from $(p_j,q_a)$ to $(f,q_b)$.  By \Cref{lem:crossing}, $(w,q_b)$ is reachable from $(p_j,q_a)$ at
threshold $\rho$.  Since $j<i$, the starting-edge clipping for
$\widetilde R_j$ cannot remove $w$,  so $w\in\widetilde R_j$.  If $w$
lay after the frontier current when $\widetilde R_j$ was processed, it was
appended at that step.  Otherwise the same argument applies with an
earlier array in place of $\widetilde R_j$.  By induction on the
processing order, $w$ is already in the union.  Thus every point of
$\widetilde R_i$ at or before the frontier is already in the union.  If the stored farthest point of
$\widetilde R_i$ is not after the current frontier, the array is discarded
in $O(1)$ time; otherwise only the newly exposed portion after the frontier
is scanned.  Since the frontier moves monotonically along $T$, every host
edge is scanned $O(1)$ times during the entire merge, so the merge takes
$O(\mu_1)$ time.

For completeness, let $z\in S[e_i]$, $z\le_T w$, and
$\dF(T[z,w],J)\le\rho$.  Since $z$ is matched to $q_a$, we have $z\in L_i$
and $\ell_i\le_T z$.  The segment $T[p_i,z]$ lies on one host edge inside
the convex ball $\mathbb B(q_a,\rho)$, so it can be matched to the
stationary point $q_a$ at threshold $\rho$; concatenating with the matching
of $T[z,w]$ to $J$ shows $w\in R(i,p_i)$.  Since $w\ge_T z\ge_T\ell_i$,
starting-edge clipping keeps $w$, and the greedy merge places it in the
output.

For soundness, let $w$ be returned from $\widetilde R_i$.  Then
$w\ge_T\ell_i$ and $\dF(T[p_i,w],J)\le\rho$.  Because $p_i$ and $\ell_i$ lie
on the same host edge and $\norm{p_i-\ell_i}\le\eta\rho$, we have
$\dF(T[p_i,w],T[\ell_i,w])\le\eta\rho$, and the triangle inequality gives
$\dF(T[\ell_i,w],J)\le(1+\eta)\rho$.  Since $\ell_i\in S$, the returned
point belongs to $\mathcal T_Q^*(J,(1+\eta)\rho,S)$.
\end{proof}

\begin{lemma}[Fr\'echet concatenation]
\label{lem:frechet-concat}
Let $A_1,\ldots,A_h$ and $B_1,\ldots,B_h$ be curves whose consecutive
endpoints agree, so that the concatenations $A_1\circ\cdots\circ A_h$ and
$B_1\circ\cdots\circ B_h$ are defined.  If $\dF(A_i,B_i)\le r$ for every
$i$, then $\dF(A_1\circ\cdots\circ A_h,\,B_1\circ\cdots\circ B_h)\le r$.
\end{lemma}

\begin{proof}
Choose a width-$r$ Fr\'echet matching for each pair $(A_i,B_i)$,
reparameterize these matchings on consecutive subintervals of $[0,1]$, and
concatenate them.  The resulting matching is monotone, and its width is the
maximum of the widths of the pieces, hence at most $r$.
\end{proof}

\begin{theorem}[Transfer through a preprocessed curve]
\label{thm:canonical-transfer}
Fix a threshold $\rho>0$.  For an auxiliary curve $Q$ of size $t$,
preprocessing the canonical intervals of a balanced binary decomposition of
$Q$ gives a data structure for this threshold in
$\widetilde O(\mu_1^2 t/\eta)$ time and space.  Given a subcurve
$\pi\subseteq Q$ and an interval array $S$, the query
$\mathsf{Transfer}_Q(\pi,\rho,S)$ takes $\widetilde O(\mu_1)$ time and
satisfies
\[
    \mathcal T_Q^*(\pi,\rho,S)
    \subseteq
    \mathsf{Transfer}_Q(\pi,\rho,S)
    \subseteq
    \mathcal T_Q^*(\pi,(1+\eta)\rho,S).
\]
\end{theorem}

\begin{proof}
Build a balanced binary decomposition tree on the edge sequence of $Q$.  A
node spanning the edges from $q_{i}q_{i+1}$ through $q_{j-1}q_j$ represents
the curve $Q[q_i,q_j]$ and stores the interval data structure from
\Cref{lem:interval-transfer}.  The tree has $O(t)$ nodes, the sum of the
lengths of all represented curves is $O(t\log t)$, and by
\Cref{lem:interval-transfer} the total preprocessing time and space are
$\sum_J O(\mu_1^2|J|/\eta)=\widetilde O(\mu_1^2 t/\eta)$.

For a query subcurve $\pi=Q[x,y]$, decompose the edge range of its maximal
vertex-to-vertex portion into $O(\log t)$ disjoint canonical intervals, in
curve order.  Their represented curves $J_1,\ldots,J_h$ share consecutive
endpoints, and $\pi=\alpha\circ J_1\circ\cdots\circ J_h\circ\gamma$, where
$\alpha$ and $\gamma$ are possibly empty partial-edge pieces.  Process these
pieces in order, each partial-edge piece is handled by the standard
free-space DP on $T$ against a single segment, and each canonical interval
is handled by its stored interval table.  This takes $\widetilde O(\mu_1)$
time.

Completeness follows by cutting any width-$\rho$ witness matching at the
piece boundaries of $\pi$ and applying completeness of each processed piece.
For soundness, each partial edge step is exact at threshold $\rho$, and each
canonical-interval step is sound at threshold $(1+\eta)\rho$.  The existential 
witnesses guaranteed by the soundness of the processed pieces concatenate, 
and \Cref{lem:frechet-concat} shows that the width remains the maximum piece width
rather than accumulating over the $O(\log t)$ intervals.  Hence the
returned array is sound at threshold $(1+\eta)\rho$.
\end{proof}

\subsection{Macro surrogates}
\label{subsec:continuous-macro-surrogates}

Fix a scale $s\in\{\mu_3,\mu_2\}$ and partition the host block $T=\tau_k$
into consecutive macros of at most $s$ edges.  All simplifications in this
subsection are endpoint-augmented.  Fix a sufficiently large constant $C_0$ and put $B_s\defeq C_0 s$.  For every macro
boundary $a$, let $T[\bar a,a]$ be the longest vertex-aligned suffix ending
at $a$ whose endpoint-augmented simplification at threshold $\delta$ has at
most $B_s$ vertices, and let $T[a,\tilde a]$ be the longest vertex-aligned
prefix starting at $a$ with the same property.  Denote the two
simplifications by $\bar\xi_a$ and $\tilde\xi_a$, respectively; they satisfy
$|\bar\xi_a|,|\tilde\xi_a|\le B_s$,
$\dF(T[\bar a,a],\bar\xi_a)\le(1+\eta)\delta$, and
$\dF(T[a,\tilde a],\tilde\xi_a)\le(1+\eta)\delta$.

For an $s$-edge macro $M=T[a,b]$, define
$Q_M\defeq\bar\xi_a\circ T[a,b]\circ\tilde\xi_b$, with common endpoints
identified.  Then $|Q_M|\le 2B_s+(s+1)=O(s)$, and concatenating the two
simplification matchings with the identity matching on $T[a,b]$ gives
\begin{equation}
\label{eq:macro-domain-QM}
    \dF(T[\bar a,\tilde b],Q_M)
    \le
    (1+\eta)\delta .
\end{equation}
A macro $M$ is \emph{marked} by a curve $Z$ if some subcurve $P\subseteq T$
with $\dF(P,Z)\le\delta$ intersects an edge of $M$, where endpoint-only
intersection is allowed.

\begin{lemma}[Macro surrogate]
\label{lem:macro-surrogate}
Let $M=T[a,b]$ be an $s$-edge macro, and let $Z$ have $O(s)$ edges.  If
$M$ is marked by $Z$, then there exists a subcurve $\pi\subseteq Q_M$ with
$\dF(\pi,Z)\le(2+\eta)\delta$.  Moreover, the standard free-space DP on
$Q_M\times Z$, initialized with all feasible bottom-boundary points, finds
such a $\pi$ at threshold $(2+\eta)\delta$ in $O(s|Z|)$ time.
\end{lemma}

\begin{proof}
Let $P=T[x,y]$ certify that $M$ is marked by $Z$.  We first prove that
$P\subseteq T[\bar a,\tilde b]$.

For the left endpoint, if $a\le_T x$, then $\bar a\le_T a\le_T x$.
Otherwise $x<_T a$.  Since $P$ intersects an edge of $M$ and endpoint-only
intersection is allowed, the subcurve from $x$ reaches the boundary point
$a$ before entering the interior of $M$.  Restrict a width-$\delta$ matching between $P$ and $Z$ to $T[x,a]$.  Let
$Z_L$ be the image subcurve and let $z_x$ be its first point.  Then
$\dF(T[x,a],Z_L)\le\delta$ and $\norm{x-z_x}\le\delta$.  Since $Z$ has
$O(s)$ edges, $|Z_L|\le|Z|+2\le C_1 s+3$ for an absolute constant $C_1$.  Let $x^-$ be $x$ if
$x$ is a vertex of $T$, and otherwise the preceding host vertex, and define
$R_L\defeq T[x^-,x]\circ\overline{xz_x}\circ Z_L$.  Matching $T[x^-,x]$
identically, traversing the connector $\overline{xz_x}$ while the host side
stays at $x$, and then following the matching between $T[x,a]$ and $Z_L$
gives $\dF(T[x^-,a],R_L)\le\delta$ with $|R_L|\le C_1 s+7$.  Therefore
$k_{\mathrm c}^*(T[x^-,a],\delta)\le C_1 s+7$.  By the continuous
simplification guarantee (\Cref{lem:simplification}),
$|\operatorname{simp}(T[x^-,a])|
\le\max\{1,2k_{\mathrm c}^*(T[x^-,a],\delta)-2\}+2\le 2C_1 s+14\le B_s$
after choosing $C_0\ge 2C_1+14$, so the
endpoint-augmented simplification of $T[x^-,a]$ is admissible for the
definition of $T[\bar a,a]$, and maximality gives
$\bar a\le_T x^-\le_T x$.

For the right endpoint, the symmetric argument gives $y\le_T\tilde b$.  If
$y\le_T b$, this is immediate.  Otherwise restrict the matching to
$T[b,y]$, let $Z_R$ be the image subcurve with last point $z_y$, take $y^+$
equal to $y$ if $y$ is a vertex and otherwise the next host vertex, and set
$R_R\defeq Z_R\circ\overline{z_y y}\circ T[y,y^+]$.  Matching the connector
while the host side stays at $y$, and then matching the added host segment
identically, gives $\dF(T[b,y^+],R_R)\le\delta$ with $|R_R|\le C_1 s+7$.  Hence
$k_{\mathrm c}^*(T[b,y^+],\delta)\le C_1 s+7$ and
$|\operatorname{simp}(T[b,y^+])|
\le\max\{1,2k_{\mathrm c}^*(T[b,y^+],\delta)-2\}+2\le 2C_1 s+14\le B_s$.
Maximality of $T[b,\tilde b]$ yields $y\le_T y^+\le_T\tilde b$.  The two
containments together prove $P\subseteq T[\bar a,\tilde b]$, including the
case where a matching cut lies in the interior of a host edge.

By \eqref{eq:macro-domain-QM} and restriction of a Fr\'echet matching to a
host subcurve, there exists $\pi\subseteq Q_M$ with
$\dF(P,\pi)\le(1+\eta)\delta$.  Combining this with the marking condition
$\dF(P,Z)\le\delta$ gives
$\dF(\pi,Z)\le\dF(\pi,P)+\dF(P,Z)\le(2+\eta)\delta$.  The stated
initialization tests all subcurves of $Q_M$ against $Z$ at once, and since
$|Q_M|=O(s)$, its running time is $O(|Q_M|\,|Z|)=O(s|Z|)$.
\end{proof}

%----------------------------------------------------------------------
\subsection{Continuous auxiliary propagation}
\label{subsec:continuous-auxiliary-propagation}
%----------------------------------------------------------------------

This subsection handles bottom-to-top propagation for one block pair.
Assume $1\le\mu_3\le\mu_2\le\mu_1$, $\mu_3\mid\mu_2$, $\mu_2\mid\mu_1$, and
$1\le\omega\le\mu_1/\mu_3$.  Split $\sigma_l$ into
$R=\lceil\mu_2/\mu_3\rceil$ consecutive pieces
$\sigma_{l,1},\ldots,\sigma_{l,R}$, each with at most $\mu_3$ edges, where
consecutive pieces share their common endpoint.  All transfer queries below
use threshold $(3+\eta)\delta$.

Use fine macros of length at most $\mu_3$.  A piece $\sigma_{l,r}$ is
\emph{dense} if it marks at least $\omega$ fine macros; there are at most
$\lceil\mu_1/\mu_3\rceil$ fine macros in total.  For a sufficiently large
constant $C$, sample $K=\lceil C(\mu_1/(\omega\mu_3))\log n\rceil$ fine
macros independently with replacement for each piece.  For each sampled fine
macro $F$, run the standard free-space DP on $Q_F\times\sigma_{l,r}$,
initialized with all feasible bottom-boundary points, at threshold
$(2+\eta)\delta$; if the computation returns a non-null subcurve, store one
such $\pi_r\subseteq Q_F$ satisfying
$\dF(\pi_r,\sigma_{l,r})\le(2+\eta)\delta$.  Let
$\mathcal E_{\mathrm{cont}}$ be the event that, over all block pairs and all
pieces, every dense piece has at least one marked sampled fine macro.

\begin{lemma}[Sampling success]
\label{lem:continuous-sampling-success}
For every fixed constant $c>0$, the hidden sampling constant can be chosen
so that $\Pr[\mathcal E_{\mathrm{cont}}]\ge 1-n^{-c}$.
\end{lemma}

\begin{proof}
For one dense piece, all $K$ samples miss its marked macros with probability
at most
$(1-\omega/\lceil\mu_1/\mu_3\rceil)^K
\le\exp(-K\omega/\lceil\mu_1/\mu_3\rceil)\le n^{-(c+3)}$
for a sufficiently large $C$.  The total number of pieces over all block
pairs is at most
$O((n/\mu_1)\cdot(m/\mu_2)\cdot(\mu_2/\mu_3))
=O(nm/(\mu_1\mu_3))\le O(n^2)$, using $m\le n$, and a union bound proves the
claim.
\end{proof}

\paragraph{Sequential case.}
Assume that every piece obtains a surrogate $\pi_r\subseteq Q_{F_r}$ with
$\dF(\pi_r,\sigma_{l,r})\le(2+\eta)\delta$.  This condition is all that the
correctness proof needs; density is used only to make sampling successful.
Let $\sigma_l^{\le r}\defeq\sigma_{l,1}\circ\cdots\circ\sigma_{l,r}$ for
$r\ge1$, and let $\sigma_l^{\le0}$ be the constant curve at $b_l$.  Define
$S_0\defeq A_B^\delta$ and
$S_r\defeq\mathsf{Transfer}_{Q_{F_r}}\bigl(\pi_r,(3+\eta)\delta,
S_{r-1}\bigr)$ for $r=1,\ldots,R$.  The raw bottom-to-top output is
$\widehat I_{\mathrm{BT}}\defeq S_R$.  The top-boundary contribution used by
the local update is its \emph{$\delta$-free lift} $I_{\mathrm{BT}}$,
consisting of all points $(y,b_{l+1})\in T_{k,l}$ with
$y\in\widehat I_{\mathrm{BT}}$ and $\norm{y-b_{l+1}}\le\delta$.

\begin{lemma}[Sequential bottom-to-top propagation]
\label{lem:aux-trans-sequential}
Assume the sequential branch is used and that $A_B$ is complete.  Then, for
every $r=0,\ldots,R$:
\begin{enumerate}
    \item If $(x,b_l)$ is globally $\delta$-reachable, $x\in A_B^\delta$,
    and $\dF(\tau_k[x,y],\sigma_l^{\le r})\le\delta$, then $y\in S_r$.
    \item If $y\in S_r$, then some $x\in A_B^\delta$ satisfies
    $\dF(\tau_k[x,y],\sigma_l^{\le r})\le(5+5\eta+\eta^2)\delta$.
\end{enumerate}
\end{lemma}

\begin{proof}
We prove both statements by induction on $r$.  For $r=0$, the curve
$\sigma_l^{\le0}$ is the constant curve at $b_l$.  If $(x,b_l)$ is globally
$\delta$-reachable and $\dF(\tau_k[x,y],\sigma_l^{\le0})\le\delta$, then
concatenating the global certificate to $(x,b_l)$ with this constant-curve
matching shows that $(y,b_l)$ is globally $\delta$-reachable; completeness
of $A_B$ gives $y\in A_B$, and the same constant-curve matching gives
$\norm{y-b_l}\le\delta$, hence $y\in A_B^\delta=S_0$.  Conversely, if
$y\in S_0$, then taking $x=y$ gives
$\dF(\tau_k[y,y],\sigma_l^{\le0})=\norm{y-b_l}\le\delta
\le(5+5\eta+\eta^2)\delta$.

Let $r\ge1$.  For completeness, cut a width-$\delta$ matching between
$\tau_k[x,y]$ and $\sigma_l^{\le r}$ at the boundary between
$\sigma_l^{\le r-1}$ and $\sigma_{l,r}$.  This gives a point $z$ with
$\dF(\tau_k[x,z],\sigma_l^{\le r-1})\le\delta$ and
$\dF(\tau_k[z,y],\sigma_{l,r})\le\delta$.  By induction, $z\in S_{r-1}$, and
since $\dF(\pi_r,\sigma_{l,r})\le(2+\eta)\delta$,
\[
    \dF(\tau_k[z,y],\pi_r)
    \le
    \dF(\tau_k[z,y],\sigma_{l,r})
    +\dF(\sigma_{l,r},\pi_r)
    \le
    (3+\eta)\delta .
\]
Completeness of \Cref{thm:canonical-transfer} places $y$ in $S_r$.

For soundness, let $y\in S_r$.  By transfer soundness, some $z\in S_{r-1}$
satisfies $\dF(\tau_k[z,y],\pi_r)\le(1+\eta)(3+\eta)\delta$, and therefore
\[
    \dF(\tau_k[z,y],\sigma_{l,r})
    \le
    (1+\eta)(3+\eta)\delta+(2+\eta)\delta
    =
    (5+5\eta+\eta^2)\delta .
\]
The induction hypothesis gives a point $x\in A_B^\delta$ whose prefix to
$z$ matches $\sigma_l^{\le r-1}$ within the same bound.  Concatenating the
two matchings and using \Cref{lem:frechet-concat} proves the claim.
\end{proof}

\paragraph{Sparse case.}
If some piece $\sigma_{l,r^*}$ fails to produce a sampled surrogate, then on
the event $\mathcal E_{\mathrm{cont}}$ it is not dense.  Run the standard
free-space DP on $\tau_k\times\sigma_{l,r^*}$ at threshold $\delta$,
initialized with all feasible bottom-boundary points; its top row contains
exactly the endpoints $y'$ for which some subcurve $\tau_k[x',y']$ satisfies
$\dF(\tau_k[x',y'],\sigma_{l,r^*})\le\delta$.  Collect every host edge with
a nonempty top-row interval; if such an endpoint is a host vertex, include
both incident host edges that exist.  Bucket the collected edges into fine
macros.  Each collected edge belongs to a fine macro marked by the failed
piece; when the endpoint is a vertex, endpoint-only intersection makes both
incident edges marked.  If at least $\omega$ distinct fine macros are
obtained, declare sampling failure and reject; on
$\mathcal E_{\mathrm{cont}}$ this does not occur.
Otherwise, map the fine macros to their containing coarse macros and
deduplicate; fewer than $\omega$ coarse macros remain.

Use coarse macros of length at most $\mu_2$, and test the candidate coarse
macros in any order, stopping at the first successful one.  For a candidate
$G$, run the standard free-space DP on $Q_G\times\sigma_l$ at threshold
$(2+\eta)\delta$, initialized with all feasible bottom-boundary points.  If
a non-null subcurve $\pi\subseteq Q_G$ is found, define the raw output
$\widehat I_{\mathrm{BT}}\defeq
\mathsf{Transfer}_{Q_G}\bigl(\pi,(3+\eta)\delta,A_B^\delta\bigr)$; if no
candidate succeeds, set $\widehat I_{\mathrm{BT}}\defeq\emptyset$.  In
either case, obtain $I_{\mathrm{BT}}$ from $\widehat I_{\mathrm{BT}}$ by
the same $\delta$-free lift as in the sequential case.

\begin{lemma}[Sparse bottom-to-top propagation]
\label{lem:aux-trans-sparse}
Assume that $A_B$ is complete.  On $\mathcal E_{\mathrm{cont}}$, if the
sparse fallback is used, then $I_{\mathrm{BT}}$ contains every true
bottom-to-top target.  If $(y,b_{l+1})\in I_{\mathrm{BT}}$, then some
$x\in A_B^\delta$ satisfies
$\dF(\tau_k[x,y],\sigma_l)\le(5+5\eta+\eta^2)\delta$.
\end{lemma}

\begin{proof}
Let $y$ be a true bottom-to-top target, so that some globally
$\delta$-reachable source $(x,b_l)$ satisfies
$\dF(\tau_k[x,y],\sigma_l)\le\delta$.  Completeness of $A_B$ and source
feasibility give $x\in A_B^\delta$.  Cut the matching at both endpoints of
the failed piece $\sigma_{l,r^*}$, the restricted host subcurve ends at a
point $y'$ collected by the enumeration, hence an incident edge containing
$y'$ is collected, its fine macro is bucketed, and the corresponding coarse
macro $G$ is tested.  The full witness also intersects $G$ and matches the
full curve $\sigma_l$ within $\delta$, so by \Cref{lem:macro-surrogate} the
test on $Q_G\times\sigma_l$ succeeds and returns a subcurve
$\pi\subseteq Q_G$ with $\dF(\pi,\sigma_l)\le(2+\eta)\delta$.  At least one
candidate therefore succeeds whenever a true target exists.

The algorithm may use any successful candidate $\pi$, not necessarily the
one associated with the particular target $y$.  This is sufficient because
every true witness satisfies
$\dF(\tau_k[x,y],\pi)\le\dF(\tau_k[x,y],\sigma_l)+\dF(\sigma_l,\pi)
\le(3+\eta)\delta$, so completeness of \Cref{thm:canonical-transfer} places
$y$ in $\widehat I_{\mathrm{BT}}$.  Since a true bottom-to-top target lies
on the $\delta$-free top boundary, $\norm{y-b_{l+1}}\le\delta$, and hence
$(y,b_{l+1})\in I_{\mathrm{BT}}$.

For soundness, let $(y,b_{l+1})\in I_{\mathrm{BT}}$.  Then
$y\in\widehat I_{\mathrm{BT}}$, and transfer soundness gives some
$x\in A_B^\delta$ with
$\dF(\tau_k[x,y],\pi)\le(1+\eta)(3+\eta)\delta$.  Therefore
$\dF(\tau_k[x,y],\sigma_l)
\le(1+\eta)(3+\eta)\delta+(2+\eta)\delta=(5+5\eta+\eta^2)\delta$.
\end{proof}

\begin{theorem}[Bottom-to-top propagation]
\label{thm:aux-trans-recurrence}
Assume that $A_B$ is complete.  On $\mathcal E_{\mathrm{cont}}$, the
bottom-to-top contribution $I_{\mathrm{BT}}$ contains every outgoing
top-boundary point with a bottom-to-top $\delta$-witness.  Moreover, if
$(y,b_{l+1})\in I_{\mathrm{BT}}$, then there exists $x\in A_B^\delta$ such
that $\dF(\tau_k[x,y],\sigma_l)\le(5+5\eta+\eta^2)\delta$.  The
per-block-pair running time, excluding host-block preprocessing, is
$\widetilde O(\mu_1\mu_3+\mu_1\mu_2/\mu_3+\mu_1\mu_2/\omega+\omega\mu_2^2)$.
\end{theorem}

\begin{proof}
If every piece obtains a surrogate, \Cref{lem:aux-trans-sequential} with
$r=R$ gives completeness and the stated local certificate for the raw set
$\widehat I_{\mathrm{BT}}$; the final $\delta$-free lift to
$I_{\mathrm{BT}}$ preserves every true top-boundary target and only removes
non-free boundary points.  Otherwise, on $\mathcal E_{\mathrm{cont}}$, the
first failed piece is not dense, so \Cref{lem:aux-trans-sparse} applies.

For the running time, the failed-piece endpoint enumeration costs
$O(\mu_1\mu_3)$.  The sequential branch makes $O(\mu_2/\mu_3)$ transfer
queries, costing $\widetilde O(\mu_1\mu_2/\mu_3)$ in total.  Fine-macro
sampling and surrogate search cost
$(\mu_2/\mu_3)\cdot(\mu_1/(\omega\mu_3))\cdot\mu_3^2=\mu_1\mu_2/\omega$, up
to logarithmic factors.  In the sparse fallback, fewer than $\omega$ coarse
macros are tested against the full block $\sigma_l$, and each test costs
$O(\mu_2^2)$, for a total of $O(\omega\mu_2^2)$.  The final transfer query
in the sparse branch is absorbed by the preceding terms.
\end{proof}

%----------------------------------------------------------------------
\subsection{Completing the continuous algorithm}
\label{subsec:continuous-completion}
%----------------------------------------------------------------------

The right-boundary contributions are
$I_{\mathrm{LR}}\cup I_{\mathrm{BR}}$, and the top-boundary contributions
are $I_{\mathrm{LT}}\cup I_{\mathrm{BT}}$.  We convert them into boundary
interval arrays by taking forward closure within the $\delta$-free space.
Let $g$ be an elementary edge of $R_{k,l}$ or $T_{k,l}$, oriented according
to the order of the varying curve, and let $F_\delta(g)$ be its
$\delta$-free interval.  For a set $J$ on the boundary, define
$\operatorname{cl}_\delta(J)$ edgewise, if $J\cap g=\emptyset$, the entry on
$g$ is empty, and otherwise
$\operatorname{cl}_\delta(J)[g]\defeq\set{z\in F_\delta(g):\ell_g\le_g z}$,
where $\ell_g$ is the first point of $J\cap g$.  The edge-assignment
convention from \Cref{subsec:continuous-dyadic-auxtrans} resolves shared
vertices.  The continuous local update returns
$A_R\defeq\operatorname{cl}_\delta(I_{\mathrm{LR}}\cup I_{\mathrm{BR}})$
and
$A_T\defeq\operatorname{cl}_\delta(I_{\mathrm{LT}}\cup I_{\mathrm{BT}})$.

\begin{lemma}[Continuous local update guarantee]
\label{lem:continuous-reach}
On $\mathcal E_{\mathrm{cont}}$, if the incoming sets $A_L$ and $A_B$ are
complete and $\beta$-sound, then the outgoing arrays $A_R$ and $A_T$ are
complete and $\beta$-sound.  After constructing the four contributions, the
assembly takes $O(\mu_1+\mu_2)$ time.
\end{lemma}

\begin{proof}
Let $z\in R_{k,l}$ be globally $\delta$-reachable.  By
\Cref{lem:continuous-witness-classification}, it has either a left-to-right
or bottom-to-right witness.  By \Cref{lem:local-boundary-propagation},
$z\in I_{\mathrm{LR}}\cup I_{\mathrm{BR}}$, and hence $z\in A_R$.
Similarly, every globally $\delta$-reachable point of $T_{k,l}$ has either a
left-to-top or bottom-to-top witness; the first case is handled by
\Cref{lem:local-boundary-propagation}, and the second by
\Cref{thm:aux-trans-recurrence}.  Thus the point belongs to
$I_{\mathrm{LT}}\cup I_{\mathrm{BT}}$ and hence to $A_T$.  Completeness
follows.

For soundness, let $z\in A_R$.  There is an anchor
$z_0\in I_{\mathrm{LR}}\cup I_{\mathrm{BR}}$ on the same elementary boundary
edge such that $z_0\le_g z$ and the boundary segment from $z_0$ to $z$ lies
in $F_\delta(g)$.  The anchor is globally $\beta\delta$-reachable by
\Cref{lem:local-boundary-propagation}.  Extend its certificate from $z_0$ to
$z$ while holding the $\tau$-position fixed at $a_{k+1}$; this extension has
cost at most $\delta$, so the complete certificate has cost at most
$\beta\delta$.

The argument for $z\in A_T$ has one additional case.  Let
$z_0\in I_{\mathrm{LT}}\cup I_{\mathrm{BT}}$ be the anchor on the same top
boundary edge.  If $z_0\in I_{\mathrm{LT}}$, then
\Cref{lem:local-boundary-propagation} gives global $\beta\delta$-reachability
of $z_0$.  If $z_0=(y_0,b_{l+1})\in I_{\mathrm{BT}}$, then
\Cref{thm:aux-trans-recurrence} gives some $x\in A_B^\delta$ with
$\dF(\tau_k[x,y_0],\sigma_l)\le(5+5\eta+\eta^2)\delta\le\beta\delta$.  Since
$A_B$ is $\beta$-sound, $(x,b_l)$ is globally $\beta\delta$-reachable, and
concatenating this global prefix certificate with the local bottom-to-top
certificate proves global $\beta\delta$-reachability of $z_0$.  Finally,
extend from $z_0$ to $z$ along the top-boundary free interval while holding
the $\sigma$-position fixed at $b_{l+1}$; this extension has cost at most
$\delta\le\beta\delta$.

The right and top boundaries contain $O(\mu_2)$ and $O(\mu_1)$ elementary
edges, respectively, so forming forward closures takes $O(\mu_1+\mu_2)$
time.
\end{proof}

\paragraph{Preprocessing.}
By \Cref{thm:canonical-transfer}, the transfer tables of one auxiliary curve
$Q$ of size $t$ are preprocessed in $\widetilde O(\mu_1^2 t/\eta)$ time.  At
the fine scale, there are $\mu_1/\mu_3$ fine macros, each with auxiliary
curve size $O(\mu_3)$, so the total fine-scale transfer preprocessing per
host block is
$(\mu_1/\mu_3)\cdot\widetilde O(\mu_1^2\mu_3/\eta)
=\widetilde O(\mu_1^3/\eta)$; at the coarse scale, there are $\mu_1/\mu_2$
coarse macros, each with auxiliary curve size $O(\mu_2)$, giving
$(\mu_1/\mu_2)\cdot\widetilde O(\mu_1^2\mu_2/\eta)
=\widetilde O(\mu_1^3/\eta)$ as well.  For one scale $s$, there are
$O(\mu_1/s)$ macro boundaries, and testing all vertex-aligned suffixes and
prefixes at one boundary has total input complexity $O(\mu_1^2)$, so the simplification work at that scale is
$\widetilde O_{d,\eps}(\mu_1^3/s)\subseteq\widetilde O_{d,\eps}(\mu_1^3)$.
The two scales, the local prefix and suffix simplifications, and recovery of
the selected matchings are therefore bounded by
$\widetilde O_{d,\eps}(\mu_1^3)$ per host block.  Since there are
$O(n/\mu_1)$ host blocks, the total continuous preprocessing is
$\widetilde O_{d,\eps}(n\mu_1^2)$.

\paragraph{Running time.}
Including the lower-order local boundary work, the total running time is
\[
    \widetilde O\!\left(
        \frac{nm}{\mu_1\mu_2}
        \left(
            \mu_1+\mu_2^2
            +
            \mu_1\mu_3
            +
            \frac{\mu_1\mu_2}{\mu_3}
            +
            \frac{\mu_1\mu_2}{\omega}
            +
            \omega\mu_2^2
        \right)
        +
        n\mu_1^2
    \right).
\]
The local boundary terms contribute $nm/\mu_2+nm\,\mu_2/\mu_1$, which will
be lower order, and the dominant expression is
\[
    \widetilde O\!\left(
        nm\left(
            \frac{\mu_3}{\mu_2}
            +
            \frac1{\mu_3}
            +
            \frac1{\omega}
            +
            \frac{\omega\mu_2}{\mu_1}
        \right)
        +
        n\mu_1^2
    \right).
\]
Choose $\mu_1=m^{4/9}$, $\mu_2=m^{2/9}$, $\mu_3=m^{1/9}$, and
$\omega=m^{1/9}$, rounded so that $\mu_3\mid\mu_2$ and $\mu_2\mid\mu_1$;
this changes each parameter by at most a constant factor.  Then the
fixed-threshold procedure runs in
$\widetilde O_{d,\eps}(nm^{8/9})$ time.

\begin{proof}[Proof of \Cref{thm:continuous-gap}]
On $\mathcal E_{\mathrm{cont}}$, \Cref{lem:continuous-reach} shows that every
local update maps complete and $\beta$-sound incoming sets to complete and
$\beta$-sound outgoing sets.  By
\Cref{lem:continuous-block-induction}, the algorithm accepts whenever
$\dF(\tau,\sigma)\le\delta$ and rejects whenever
$\dF(\tau,\sigma)>\beta\delta=(5+\eps)\delta$.
By \Cref{lem:continuous-sampling-success}, the event
$\mathcal E_{\mathrm{cont}}$ holds with high probability.  The explicit
sampling-failure cutoff bounds the running time even outside this event, so
the runtime is $\widetilde O_{d,\eps}(nm^{8/9})$.
\end{proof}

\begin{theorem}[Continuous Fr\'echet approximation]
\label{thm:continuous-main}
For two polygonal curves $\tau$ and $\sigma$ in fixed dimension, with
$\abs{\tau}=n$ and $\abs{\sigma}=m\le n$, one can compute a randomized
$(5+\eps)$-approximation to their continuous Fr\'echet distance, with high
probability, in
\[
    \widetilde O_{d,\eps}(nm^{8/9})
\]
time.
\end{theorem}

\begin{proof}
Colombe and Fox~\cite{colombe2021approximating} convert any
$\alpha$-approximate decision procedure into a
$(1+\gamma)\alpha$-approximation for any $\gamma\in(0,1]$.  The
transformation makes $O(\log n)$ decision calls and spends $O(n\log n)$
additional time.  It queries the decision procedure only through its
accept and reject answers, so its value version applies to our
gap-decision procedure.  Run \Cref{thm:continuous-gap} with accuracy
parameter $\eps_0=\eps/2$ and apply the transformation with
$\gamma=\eps/12$; since $\eps\le1/2$,
$(1+\eps/12)(5+\eps/2)\le 5+\eps$.  Choosing the sampling constant so
that a union bound over the $O(\log n)$ calls preserves high
probability, the total running time is
$\widetilde O_{d,\eps}(nm^{8/9})$.
\end{proof}

%% file: discrete.tex
\section{Discrete Fr\'echet Algorithm}
\label{sec:discrete}
We now prove the discrete result through a fixed-threshold gap-decision
procedure.  The approximation algorithm for the discrete Fr\'echet distance is given
in \Cref{subsec:discrete-completion}.

\begin{theorem}[Discrete gap decision]
\label{thm:discrete-gap}
Let $\tau$ and $\sigma$ be polygonal curves in fixed dimension, with
$\abs{\tau}=n$ and $\abs{\sigma}=m\le n$.  Given a threshold $\delta>0$, 
there is a randomized algorithm that accepts if
$\ddF(\tau,\sigma)\le\delta$ and rejects if
$\ddF(\tau,\sigma)>5\delta$, with high probability.  The algorithm may
return either answer in the remaining gap.  Its running time is
$\widetilde O_d(nm^{4/5})$.
\end{theorem}

%----------------------------------------------------------------------
\subsection{Fixed-threshold block dynamic program}
\label{subsec:discrete-block-dp}
%----------------------------------------------------------------------

The algorithm processes the block pairs $\tau_k\times\sigma_l$ in any
topological order of the block grid.  For each block, the discrete
implementation of \REACH{} receives subsets of the incoming boundary vertices
$L_{k,l}\cup B_{k,l}$ and returns subsets of
$R_{k,l}\cup T_{k,l}$.  A boundary shared by two adjacent blocks is stored
once and updated by taking unions.

On the two outer incoming boundaries $\{v_1\}\times\sigma$ and
$\tau\times\{w_1\}$, the algorithm stores exactly the vertex pairs reachable
from $(v_1,w_1)$ at threshold $\delta$.  These sets are computed by the
one-dimensional discrete dynamic program.  Every other boundary set is created
when first produced by a block update.  After all blocks have been processed,
the algorithm accepts exactly when $(v_n,w_m)$ is stored.

\begin{lemma}[Discrete block induction]
\label{lem:discrete-block-induction}
On any event on which every discrete \REACH{} call maps complete and
$5$-sound incoming sets to complete and $5$-sound outgoing sets, the
block dynamic program accepts whenever $\ddF(\tau,\sigma)\le\delta$ and rejects
whenever $\ddF(\tau,\sigma)>5\delta$.
\end{lemma}

\begin{proof}
We induct over the block order.  The invariant is that every boundary set
already initialized or produced is complete and $5$-sound, it holds initially because the two outer incoming boundaries are
computed exactly at threshold $\delta$.  When a block
$\tau_k\times\sigma_l$ is processed, its incoming boundary sets have already
been initialized or produced and therefore satisfy the invariant.  The assumed
local guarantee implies that its outgoing sets are complete and
$5$-sound.  Taking unions on shared boundaries preserves both properties.

If $\ddF(\tau,\sigma)\le\delta$, a monotone path in the discrete free-space
graph reaches $(v_n,w_m)$.  Completeness ensures that each boundary vertex
visited by this path is retained, and hence the algorithm accepts.
Conversely, if the algorithm accepts, then $5$-soundness of the final set
implies $\ddF(\tau,\sigma)\le5\delta$.  The factor $5$ does not
accumulate across blocks, the local discrete matchings concatenate
monotonically, and the cost of the concatenation is the maximum cost of its
pieces.
\end{proof}

%----------------------------------------------------------------------
\subsection{One-block witnesses and local cases}
\label{subsec:discrete-witness-local}
%----------------------------------------------------------------------

Fix a block pair $\tau_k\times\sigma_l$, with incoming sets
$A_L\subseteq L_{k,l}$ and $A_B\subseteq B_{k,l}$.  A
\emph{local discrete $\delta$-witness} for
$z\in R_{k,l}\cup T_{k,l}$ is a directed path at threshold $\delta$ in the
discrete free-space graph of the block, from a globally
$\delta$-reachable incoming vertex to $z$.

The four witness types are determined by their incoming and outgoing sides:
\[
\begin{array}{c|c|c}
\text{type} & \text{source and target} & \text{curve condition}\\ \hline
1
& (a_k,p)\to(a_{k+1},q)
& \ddF(\tau_k,\sigma_l[p,q])\le\delta
\\[1mm]
2
& (x,b_l)\to(a_{k+1},q)
& \ddF(\tau_k[x,a_{k+1}],\sigma_l[b_l,q])\le\delta
\\[1mm]
3
& (a_k,p)\to(y,b_{l+1})
& \ddF(\tau_k[a_k,y],\sigma_l[p,b_{l+1}])\le\delta
\\[1mm]
4
& (x,b_l)\to(y,b_{l+1})
& \ddF(\tau_k[x,y],\sigma_l)\le\delta .
\end{array}
\]
All points in this subsection are vertices, and the source in each row is
required to be globally $\delta$-reachable.  Types~1 and~2 contribute to
$R_{k,l}$, while Types~3 and~4 contribute to $T_{k,l}$.

\begin{lemma}[Discrete witness classification]
\label{lem:discrete-witness-classification}
Every globally $\delta$-reachable vertex of
$R_{k,l}\cup T_{k,l}$ has a local witness of one of the four types above.
\end{lemma}

\begin{proof}
Let $z\in R_{k,l}\cup T_{k,l}$ be globally $\delta$-reachable, and fix a
monotone grid path from $(v_1,w_1)$ to $z$.  Let $s$ be the first vertex of
the maximal suffix of this path contained in the block.  Monotonicity implies
that $s\in L_{k,l}\cup B_{k,l}$; entry through an outgoing side is possible
only at a corner that also lies on an incoming side.  The subpath from $s$ to $z$ is a local witness.  Its incoming and outgoing
sides place it in one of the four types.  At the shared corner vertex
$(a_{k+1},b_{l+1})$ of $R_{k,l}$ and $T_{k,l}$, the Type~4 and Type~2
curve conditions coincide, and so do the Type~3 and Type~1 conditions.  A
witness ending at this corner therefore has a type on each side.
\end{proof}

We construct the contributions $I_1$, $I_2$, and $I_3$ for the first three
types.  Identify $A_L$ and $A_B$ with their projections onto $\sigma_l$ and
$\tau_k$, and clip them to the $\delta$-feasible portions.  Define
$A_L^\delta\defeq\set{p\in A_L:\norm{a_k-p}\le\delta}$ and
$A_B^\delta\defeq\set{x\in A_B:\norm{x-b_l}\le\delta}$; this clipping
removes no source of a true $\delta$-witness.

For every simplification used below, add the endpoints of the original
subcurve if necessary; this increases its size by at most two and preserves
the discrete Fr\'echet bound, since while moving to or from an added
endpoint, the original curve remains fixed at the corresponding endpoint.
Let $M_P$ denote the monotone vertex selection from the stored matching.  For
every vertex subcurve $P[x,y]$,
$\ddF\bigl(P[x,y],\zeta[M_P(x),M_P(y)]\bigr)\le\delta$.

Fix a sufficiently large constant $C$, and preprocess each host block
$\tau_k$ using \Cref{lem:batched-discrete-simplifications}.  Among the stored
simplifications, retain the simplification $\zeta_k^{\mathrm{all}}$ of
$\tau_k$, provided its endpoint-augmented size is at most $C\mu_2$; a
simplification $\zeta_k^{\mathrm{suf}}$ of the longest suffix
$P_k^{\mathrm{suf}}=\tau_k[s_k,a_{k+1}]$ whose endpoint-augmented size is at
most $C\mu_2$; and a simplification $\zeta_k^{\mathrm{pre}}$ of the longest
prefix $P_k^{\mathrm{pre}}=\tau_k[a_k,t_k]$ whose endpoint-augmented size is
at most $C\mu_2$.  We materialize only these three simplifications and their
matching selections $M_k^{\mathrm{suf}}$ and $M_k^{\mathrm{pre}}$; their
preprocessing cost is accounted for in \Cref{subsec:discrete-completion}.

Set $r_{\mathrm{loc}}\defeq2\delta$.  For Type~1, run the discrete
free-space dynamic program on $\zeta_k^{\mathrm{all}}\times\sigma_l$ at
threshold $r_{\mathrm{loc}}$, with sources $A_L^\delta$ on the left boundary,
and let $W_1\subseteq\sigma_l$ be the reachable vertices on the right
boundary; if $\zeta_k^{\mathrm{all}}$ was not retained, the output is empty.
For Type~2, define
$S_{\mathrm{suf}}\defeq M_k^{\mathrm{suf}}(A_B^\delta\cap
P_k^{\mathrm{suf}})$, run the dynamic program on
$\zeta_k^{\mathrm{suf}}\times\sigma_l$ with $S_{\mathrm{suf}}$ on the bottom
boundary, and let $W_2\subseteq\sigma_l$ be the reachable vertices on the
right boundary.  For Type~3, run the dynamic program on
$\zeta_k^{\mathrm{pre}}\times\sigma_l$ with $A_L^\delta$ on the left
boundary, and let $W_3\subseteq\zeta_k^{\mathrm{pre}}$ be the reachable
vertices on the top boundary.  All three computations use threshold
$r_{\mathrm{loc}}$ and no sources on the other incoming boundary.  The
initial sets are feasible at threshold $r_{\mathrm{loc}}$, for Types~1
and~3 this follows from endpoint preservation, and for Type~2, if
$\bar x=M_k^{\mathrm{suf}}(x)$ and $x\in A_B^\delta$, then
$\norm{\bar x-b_l}\le\delta+\delta=r_{\mathrm{loc}}$.  Define
\[
\begin{aligned}
    I_1
    &\defeq
    \set{(a_{k+1},q)\in R_{k,l}:
          q\in W_1,\ \norm{a_{k+1}-q}\le\delta},\\
    I_2
    &\defeq
    \set{(a_{k+1},q)\in R_{k,l}:
          q\in W_2,\ \norm{a_{k+1}-q}\le\delta},\\
    I_3
    &\defeq
    \set{(y,b_{l+1})\in T_{k,l}:
          y\in(M_k^{\mathrm{pre}})^{-1}(W_3),\
          \norm{y-b_{l+1}}\le\delta}.
\end{aligned}
\]

\begin{lemma}[Discrete local cases]
\label{lem:discrete-types-123}
Assume that $A_L$ and $A_B$ are complete and $5$-sound.  Then
$I_1\cup I_2\cup I_3$ contains every outgoing vertex with a Type~1, Type~2,
or Type~3 $\delta$-witness.  Every vertex in this union is globally
$5\delta$-reachable.  After host-block preprocessing, the construction
takes $O(\mu_1+\mu_2^2)$ time.
\end{lemma}

\begin{proof}
For Type~1, let $(a_k,p)\to(a_{k+1},q)$ be a witness.  Completeness gives
$p\in A_L^\delta$, and $\ddF(\tau_k,\sigma_l[p,q])\le\delta$.  Since
$\sigma_l[p,q]$ has $O(\mu_2)$ vertices,
$k_{\mathrm d}^*(\tau_k,\delta)\le\abs{\sigma_l[p,q]}=O(\mu_2)$ by
\Cref{cor:kstar-witness}, so by
\Cref{lem:batched-discrete-simplifications} the stored simplification of
$\tau_k$ has $O(\mu_2)$ vertices and $\zeta_k^{\mathrm{all}}$ is retained.
The triangle inequality gives
$\ddF(\zeta_k^{\mathrm{all}},\sigma_l[p,q])\le2\delta$, so the dynamic
program reaches $q$ and $(a_{k+1},q)\in I_1$.

For Type~2, let the witness start at $(x,b_l)$.  The suffix
$\tau_k[x,a_{k+1}]$ is within distance $\delta$ of the $O(\mu_2)$-vertex
curve $\sigma_l[b_l,q]$, so by \Cref{cor:kstar-witness} its stored batched
simplification has $O(\mu_2)$ vertices, and by maximality of $P_k^{\mathrm{suf}}$ we have
$x\in P_k^{\mathrm{suf}}$.  For $\bar x=M_k^{\mathrm{suf}}(x)$, the
restricted matching and the witness imply
\[
    \ddF\bigl(
        \zeta_k^{\mathrm{suf}}
        [\bar x,\operatorname{last}(\zeta_k^{\mathrm{suf}})],
        \sigma_l[b_l,q]
    \bigr)
    \le2\delta ,
\]
so the dynamic program reaches $q$ and $(a_{k+1},q)\in I_2$.

The Type~3 argument is symmetric, a witness ending at $(y,b_{l+1})$ implies
that the stored simplification of $\tau_k[a_k,y]$ has $O(\mu_2)$ vertices,
hence $y\in P_k^{\mathrm{pre}}$ by maximality, and the restricted matching
shows that $M_k^{\mathrm{pre}}(y)\in W_3$, so $(y,b_{l+1})\in I_3$.

For soundness, every inserted vertex has a certificate at threshold
$2\delta$ in a simplified grid.  Composing it with the corresponding
stored matching gives a local discrete matching in the original block of cost
at most $3\delta$.  Its source belongs to $A_L$ or $A_B$ and is
globally $5\delta$-reachable.

Each local dynamic program has size $O(\mu_2^2)$.  Images and preimages under
the stored matchings are computed by linear scans in $O(\mu_1)$ total time.
\end{proof}

\subsection{Exact discrete auxiliary transfer}
\label{subsec:discrete-dyadic-transfer}

Fix a host block $T=(z_1,\ldots,z_N)$ with $N=O(\mu_1)$, an auxiliary curve
$Q=(q_1,\ldots,q_t)$, and a threshold $r>0$.  For a vertex subcurve
$\pi\subseteq Q$ and a source set $S\subseteq T$, define
\[
    \operatorname{Trans}_Q^*(\pi,r,S)
    \defeq
    \set{z_j\in T: \exists z_i\in S,
        \ i\le j,
        \ \ddF(T[i,j],\pi)\le r}.
\]

\begin{lemma}[Fixed-interval exact transfer]
\label{lem:disc-fixed-transfer}
Let $J=Q[a,b]$ be a fixed vertex-to-vertex subcurve of $Q$.  There is a data
structure with preprocessing time $O(\mu_1 |J|)$ and query time $O(\mu_1)$
which, for every source set $S\subseteq T$, returns exactly
$\operatorname{Trans}_Q^*(J,r,S)$.
\end{lemma}

\begin{proof}
Let $G_J$ be the discrete free-space graph of $T\times J$ at threshold $r$.
It has $O(\mu_1|J|)$ vertices and edges and is a planar directed acyclic
graph.  Build the planar reachability oracle of
Theorem~\ref{thm:planar-reachability}
on $G_J$.  For every feasible vertex $u$ of $G_J$, compute
$F(u)\defeq\max\set{j: u \leadsto (z_j,q_b)\text{ in }G_J}$, with value $0$
if no final-row vertex is reachable; this is computed by one reverse
topological scan.  The total preprocessing time is $O(\mu_1|J|)$.

Given $S$, scan the host vertices in increasing order, maintaining the
largest final-row index $f$ already exposed.  For a source $z_i\in S$, skip
it if $(z_i,q_a)$ is infeasible.  Otherwise set $g=F(z_i,q_a)$.  If $g>f$,
test by the reachability oracle all final-row vertices
$(z_{f+1},q_b),\ldots,(z_g,q_b)$, insert exactly those reachable from
$(z_i,q_a)$, and set $f:=g$. Infeasible final-row pairs are not vertices of $G_J$ and count as
unreachable. Since $f$ only increases, the total number of
oracle tests is $O(\mu_1)$.

Soundness is immediate, because a vertex is inserted only after a positive
reachability query.  For completeness, suppose a processed source $z_i$
reaches $(z_j,q_b)$.  If $j>f$ when $z_i$ is processed, then
$j\le F(z_i,q_a)$, so $z_j$ is tested and inserted.  If $j\le f$, let $z_h$
be a previously processed source that last advanced the frontier to $f$.
Since any path from $z_i$ to $z_j$ has $i\le j$, we have
$h\le i\le j\le f$, and the crossing lemma (\Cref{lem:crossing}) in the
discrete free-space graph implies that $(z_j,q_b)$ is reachable from
$(z_h,q_a)$.  If $z_j$ lay in the range tested when $z_h$ advanced the frontier, it was
inserted at that step, since it is reachable from $(z_h,q_a)$.  Otherwise
the same argument applies with an earlier source in place of $z_h$.  By
induction on the processing order, $z_j$ was already inserted.  So
every reachable final-row vertex is returned.
\end{proof}

\begin{theorem}[Dyadic exact auxiliary transfer]
\label{thm:disc-dyadic-transfer}
For an auxiliary curve $Q$ of size $t$, one can preprocess $Q$ in
$\widetilde O(\mu_1 t)$ time and space so that every query $\pi\subseteq Q$
is answered exactly in $\widetilde O(\mu_1)$ time.
\end{theorem}

\begin{proof}
Build a balanced binary decomposition tree on the edge sequence of $Q$.  A
node whose edge interval runs from $q_iq_{i+1}$ through $q_{j-1}q_j$ stores
the fixed-interval structure for $Q[i,j]$.  The total length of all stored
canonical intervals is $O(t\log t)$, so \Cref{lem:disc-fixed-transfer}
gives preprocessing time $\widetilde O(\mu_1t)$.

A query vertex subcurve $\pi=Q[a,b]$ is decomposed into $O(\log t)$ canonical
vertex subcurves in order, and their exact transfer structures are applied
one after another, starting from $S$.  If $a=b$, answer the singleton query
by one linear dynamic-programming scan on $T\times\{q_a\}$.  Each
fixed-interval query costs $O(\mu_1)$, so the total query time is
$\widetilde O(\mu_1)$.  Exactness follows by cutting a discrete matching at the shared endpoints
of the canonical subcurves.  Conversely, exact membership in each
fixed-interval transfer set yields a discrete path for each canonical
subcurve.  These paths concatenate at the shared endpoints.
\end{proof}

%----------------------------------------------------------------------
\subsection{Discrete macro surrogates}
\label{subsec:discrete-macro-surrogates}
%----------------------------------------------------------------------

We use two aligned macro scales, $s=\mu_3$ and $s=\mu_2$.  By rounding the
parameters by constant factors, assume $\mu_3\mid\mu_2\mid\mu_1$; thus every
fine macro lies in a unique coarse macro.  A macro has $O(s)$ consecutive
host vertices, and boundary vertices are assigned to the later macro.

For every macro boundary vertex $a$, let $T[\bar a,a]$ be the longest
vertex-aligned suffix ending at $a$ whose endpoint-augmented simplification
from \Cref{lem:batched-discrete-simplifications} has at most $C_0s$
vertices, for one fixed sufficiently large constant $C_0$, and let
$T[a,\tilde a]$ be the analogous longest prefix with the same size
threshold.  Denote the two simplifications by $\bar\xi_a$ and
$\tilde\xi_a$.  For a macro $M=T[a,b]$, define
$Q_M\defeq\bar\xi_a\circ T[a,b]\circ\tilde\xi_b$.  Then $|Q_M|=O(s)$ and
$\ddF(T[\bar a,\tilde b],Q_M)\le\delta$.  A macro $M$ is
\emph{marked} by a vertex curve $Z$ if there is a vertex subcurve
$P\subseteq T$ containing a host vertex assigned to $M$ such that
$\ddF(P,Z)\le\delta$.

\begin{lemma}[Discrete macro surrogate]
\label{lem:disc-macro-surrogate}
Let $M=T[a,b]$ be an $s$-scale macro and let $Z$ have $O(s)$ vertices.  If
$M$ is marked by $Z$, then there exists a vertex subcurve $\pi\subseteq Q_M$
with $\ddF(\pi,Z)\le2\delta$.  Moreover, one free-start/free-end
discrete propagation on $Q_M\times Z$ finds such a $\pi$ in $O(s|Z|)$ time.
\end{lemma}

\begin{proof}
Let $P=T[x,y]$ mark $M$.  We first prove that
$P\subseteq T[\bar a,\tilde b]$.  If $x\ge a$, then $x\ge\bar a$.  If
$x<a$, then the subcurve $P$ reaches the boundary vertex $a$.  Restrict a
width-$\delta$ discrete matching between $P$ and $Z$ to $T[x,a]$; its image
is a contiguous vertex subcurve of $Z$, hence has at most $C_1s$ vertices
for an absolute constant $C_1$, and therefore
$k_{\mathrm d}^*(T[x,a],\delta)\le C_1s$ by \Cref{cor:kstar-witness}.
\Cref{lem:batched-discrete-simplifications} and endpoint augmentation give a
simplification of $T[x,a]$ with at most $C_0s$ vertices, after choosing
$C_0$ large enough.  Thus $T[x,a]$ is admissible in the definition of
$T[\bar a,a]$, so by maximality $\bar a\le x$.  The right side is symmetric
and gives $y\le\tilde b$, proving the containment.

Restrict the stored matching between $T[\bar a,\tilde b]$ and $Q_M$ to $P$.
The image is a vertex subcurve $\pi\subseteq Q_M$ with
$\ddF(P,\pi)\le\delta$, and together with $\ddF(P,Z)\le\delta$ this
gives $\ddF(\pi,Z)\le\ddF(\pi,P)+\ddF(P,Z)\le2\delta$.  The
free-start/free-end propagation tests all subcurves of $Q_M$ against $Z$,
and its cost is $O(|Q_M||Z|)=O(s|Z|)$.
\end{proof}

%----------------------------------------------------------------------
\subsection{Discrete auxiliary propagation}
\label{subsec:discrete-macro-auxtransfer}
%----------------------------------------------------------------------

Assume that the incoming bottom-boundary set $A_B$ is complete and
$5$-sound, and set $\rho_0\defeq3\delta$.  Split $\sigma_l$ into
$R=O(\mu_2/\mu_3)$ consecutive pieces
$\sigma_{l,1},\ldots,\sigma_{l,R}$, each with $O(\mu_3)$ vertices.

Use fine macros at scale $\mu_3$.  A piece $\sigma_{l,r}$ is \emph{dense} if
it marks at least $\omega$ fine macros; let $N_F=\Theta(\mu_1/\mu_3)$ be the
number of fine macros.  For each piece, sample
$K=\lceil C(N_F/\omega)\log n\rceil$ fine macros independently with
replacement.  For each sampled fine macro $F$, run the search of
\Cref{lem:disc-macro-surrogate} on $Q_F\times\sigma_{l,r}$; if it succeeds,
store one returned subcurve $\pi_r\subseteq Q_F$.  Let
$\mathcal E_{\mathrm{disc}}$ be the event that every dense piece, over all
block pairs, receives at least one marked sampled fine macro.

\begin{lemma}[Sampling success]
\label{lem:disc-macro-sampling}
For every fixed $c>0$, the constant $C$ can be chosen so that
$\Pr[\mathcal E_{\mathrm{disc}}]\ge 1-n^{-c}$.
\end{lemma}

\begin{proof}
For one dense piece, all $K$ samples miss its marked macros with probability
at most $(1-\omega/N_F)^K\le\exp(-K\omega/N_F)\le n^{-(c+3)}$ for a
sufficiently large $C$.  The number of pieces over all block pairs is
$O((n/\mu_1)\cdot(m/\mu_2)\cdot(\mu_2/\mu_3))
=O(nm/(\mu_1\mu_3))\le O(n^2)$, because $m\le n$, and a union bound proves
the claim.
\end{proof}

\paragraph{Sequential branch.}
If every piece has a surrogate $\pi_r$ with
$\ddF(\pi_r,\sigma_{l,r})\le2\delta$, set $S_0\defeq A_B^\delta$ and
$S_r\defeq\operatorname{Trans}^*_{Q_{F_r}}(\pi_r,\rho_0,S_{r-1})$ for
$r=1,\ldots,R$, where $F_r$ is the sampled fine macro whose search produced
$\pi_r$.  Let $\sigma_l^{\le r}\defeq\sigma_{l,1}\circ\cdots\circ\sigma_{l,r}$
and let $\sigma_l^{\le0}$ be the constant curve at $b_l$.

\begin{lemma}[Sequential auxiliary transfer]
\label{lem:disc-sequential-transfer}
Assume that the sequential branch is used and that $A_B$ is complete.
For every $r=0,\ldots,R$:
\begin{enumerate}
    \item If $(x,b_l)$ is globally $\delta$-reachable, $x\in A_B^\delta$,
    and $\ddF(T[x,y],\sigma_l^{\le r})\le\delta$, then $y\in S_r$.
    \item If $y\in S_r$, then some $x\in A_B^\delta$ satisfies
    $\ddF(T[x,y],\sigma_l^{\le r})\le5\delta$.
\end{enumerate}
\end{lemma}

\begin{proof}
For $r=0$, the first claim follows because a constant matching to $b_l$ makes
$(y,b_l)$ globally $\delta$-reachable; completeness of $A_B$ and feasibility
give $y\in A_B^\delta=S_0$.  The second claim follows by taking $x=y$.

Let $r\ge1$.  For completeness, cut a width-$\delta$ matching between
$T[x,y]$ and $\sigma_l^{\le r}$ at the boundary before $\sigma_{l,r}$.  This
gives a vertex $z$ with $\ddF(T[x,z],\sigma_l^{\le r-1})\le\delta$ and
$\ddF(T[z,y],\sigma_{l,r})\le\delta$.  By induction, $z\in S_{r-1}$, and
since $\ddF(\pi_r,\sigma_{l,r})\le2\delta$, we get
$\ddF(T[z,y],\pi_r)\le3\delta=\rho_0$.  Exact transfer completeness
places $y$ in $S_r$.

For soundness, if $y\in S_r$, exact transfer gives some $z\in S_{r-1}$ with
$\ddF(T[z,y],\pi_r)\le\rho_0$, and thus
\[
    \ddF(T[z,y],\sigma_{l,r})
    \le
    5\delta .
\]
Concatenate this certificate with the induction certificate for the prefix;
the discrete Fr\'echet cost of the concatenation is the maximum of the two
piece costs, so the bound remains $5\delta$.
\end{proof}

\begin{lemma}[Marked-vertex enumeration]
\label{lem:disc-marked-vertex-enum}
For any vertex curve $Z$, all host vertices that are contained in some
subcurve $P\subseteq T$ with $\ddF(P,Z)\le\delta$ can be enumerated in
$O(\mu_1|Z|)$ time.
\end{lemma}

\begin{proof}
Run the discrete free-start/free-end dynamic program on $T\times Z$ at
threshold $\delta$, initialized from all feasible vertices on the first row.
Run the same computation backward from all feasible vertices on the last row.
A host vertex $z_i$ is contained in such a witness subcurve if and only if for
some vertex $q_j$ of $Z$, the grid vertex $(z_i,q_j)$ is both forward- and
backward-reachable.  Scanning the $O(\mu_1|Z|)$ grid vertices therefore
enumerates exactly the marked host vertices.
\end{proof}

\paragraph{Sparse fallback.}
If some piece $\sigma_{l,r^*}$ fails to produce a sampled surrogate, run
\Cref{lem:disc-marked-vertex-enum} on $Z=\sigma_{l,r^*}$, and bucket the
enumerated vertices into fine macros.  If at least $\omega$ fine macros are
obtained, declare sampling failure and reject; on
$\mathcal E_{\mathrm{disc}}$ this never happens.  Otherwise, map the
obtained fine macros to their containing coarse macros and remove
duplicates; fewer than $\omega$ coarse macros remain.  For each candidate
coarse macro $G$, run the search of \Cref{lem:disc-macro-surrogate} on
$Q_G\times\sigma_l$, stopping at the first non-null output
$\pi\subseteq Q_G$.  If none succeeds, set
$\widehat I_{\mathrm{AT}}\defeq\emptyset$, otherwise set
$\widehat I_{\mathrm{AT}}\defeq
\operatorname{Trans}^*_{Q_G}(\pi,\rho_0,A_B^\delta)$.

\begin{lemma}[Sparse auxiliary transfer]
\label{lem:disc-sparse-transfer}
Assume that $A_B$ is complete, that $\mathcal E_{\mathrm{disc}}$ holds, and
that the sparse fallback is used.  If there is a true auxiliary-transfer
witness $\ddF(T[x,y],\sigma_l)\le\delta$ with $(x,b_l)$ globally
$\delta$-reachable, then $y\in\widehat I_{\mathrm{AT}}$.  Moreover, every
$y\in\widehat I_{\mathrm{AT}}$ has some $x\in A_B^\delta$ with
$\ddF(T[x,y],\sigma_l)\le5\delta$.
\end{lemma}

\begin{proof}
Completeness of $A_B$ gives $x\in A_B^\delta$.  Restrict the witness matching
to the failed piece $\sigma_{l,r^*}$, then the restricted host subcurve contains a
host vertex that is enumerated, so its fine macro is bucketed and the
containing coarse macro $G$ is tested.  The full witness also marks $G$ by
the whole curve $\sigma_l$, hence \Cref{lem:disc-macro-surrogate} returns a
subcurve $\pi\subseteq Q_G$ with $\ddF(\pi,\sigma_l)\le2\delta$.

The algorithm may use any successful candidate $\pi$.  For every true
witness,
$\ddF(T[x,y],\pi)\le\ddF(T[x,y],\sigma_l)+\ddF(\sigma_l,\pi)
\le3\delta=\rho_0$, so exact transfer completeness inserts $y$ into
$\widehat I_{\mathrm{AT}}$.  For soundness, if
$y\in\widehat I_{\mathrm{AT}}$, exact transfer gives $x\in A_B^\delta$ with
$\ddF(T[x,y],\pi)\le3\delta$, and using
$\ddF(\pi,\sigma_l)\le2\delta$ we conclude
$\ddF(T[x,y],\sigma_l)\le5\delta$.
\end{proof}

Define the top-boundary auxiliary-transfer contribution $I_{\mathrm{AT}}$ by
the final feasible clip, $I_{\mathrm{AT}}$ consists of all vertices
$(y,b_{l+1})\in T_{k,l}$ with $y\in\widehat I_{\mathrm{AT}}$ and
$\norm{y-b_{l+1}}\le\delta$, in the sequential branch, take
$\widehat I_{\mathrm{AT}}\defeq S_R$.

\begin{theorem}[Macro auxiliary-transfer recurrence]
\label{thm:disc-macro-transfer-recurrence}
Assume that $A_B$ is complete and $5$-sound.  On
$\mathcal E_{\mathrm{disc}}$, $I_{\mathrm{AT}}$ contains every outgoing
vertex with an auxiliary-transfer $\delta$-witness, and every vertex in
$I_{\mathrm{AT}}$ is globally $5\delta$-reachable, the per-block-pair time is
$\widetilde O(\mu_1\mu_3+\mu_1\mu_2/\mu_3+\mu_1\mu_2/\omega+\omega\mu_2^2)$.
\end{theorem}

\begin{proof}
If every piece obtains a surrogate, completeness and local soundness follow
from \Cref{lem:disc-sequential-transfer} with $r=R$.  If some piece fails,
then on $\mathcal E_{\mathrm{disc}}$ it is not dense, so after the explicit
cutoff the sparse fallback has fewer than $\omega$ candidate coarse macros,
and \Cref{lem:disc-sparse-transfer} applies.  In either branch, the final
feasible clip preserves every true target, since a Type~4 witness matches
$y$ to $b_{l+1}$ within distance $\delta$.  Global soundness follows by
concatenating the local certificate from $x\in A_B^\delta$ with the
$5$-sound global certificate for $(x,b_l)$.

For the running time, the failed-piece enumeration costs $O(\mu_1\mu_3)$.
The sequential branch has $O(\mu_2/\mu_3)$ exact transfer queries, costing
$\widetilde O(\mu_1\mu_2/\mu_3)$.  Sampling and fine-scale surrogate search
cost $(\mu_2/\mu_3)\cdot(\mu_1/(\omega\mu_3))\cdot\mu_3^2=\mu_1\mu_2/\omega$,
up to logarithmic factors.  The sparse fallback tests fewer than $\omega$
coarse macros against the full block, and each test costs $O(\mu_2^2)$. The final exact transfer query in the sparse branch costs
$\widetilde O(\mu_1)$ and is absorbed by the enumeration term.
\end{proof}

%----------------------------------------------------------------------
\subsection{Completing the discrete algorithm}
\label{subsec:discrete-completion}
%----------------------------------------------------------------------

The outgoing boundary sets are $A_R\defeq I_1\cup I_2$ and
$A_T\defeq I_3\cup I_{\mathrm{AT}}$.  Together with the local cases from
\Cref{subsec:discrete-witness-local} and
\Cref{thm:disc-macro-transfer-recurrence}, the discrete \REACH{} update maps
complete and $5$-sound incoming sets to complete and $5$-sound
outgoing sets on $\mathcal E_{\mathrm{disc}}$.

We now give the running time.  The local cases cost
$\widetilde O(\mu_1+\mu_2^2)$ per block pair, and the auxiliary-transfer
contribution is the bound in \Cref{thm:disc-macro-transfer-recurrence}, so
the per-block-pair work is
$\widetilde O(\mu_1+\mu_2^2+\mu_1\mu_3+\mu_1\mu_2/\mu_3
+\mu_1\mu_2/\omega+\omega\mu_2^2)$.

Preprocessing is per host block.
\Cref{lem:batched-discrete-simplifications} costs
$\widetilde O_d(\mu_1^2)$ per host block. Selecting the longest admissible suffix or prefix at a macro boundary
takes $O(\mu_1)$ lookups, each in $O(1)$ time by
\Cref{lem:batched-discrete-simplifications}.  Over both scales and the
three simplifications of \Cref{subsec:discrete-witness-local}, all
selections cost $O(\mu_1^2/\mu_3)$ per host block, and materializing them
costs $O(\mu_1)$ per scale.  Both costs are dominated. For one macro auxiliary curve
of size $O(s)$, \Cref{thm:disc-dyadic-transfer} gives preprocessing
$\widetilde O(\mu_1 s)$, so the fine scale contributes
$(\mu_1/\mu_3) \cdot\widetilde O(\mu_1\mu_3)=\widetilde O(\mu_1^2)$ per host
block, and the coarse scale contributes the same bound.  Thus the total
preprocessing over all host blocks is
$(n/\mu_1)\cdot\widetilde O_d(\mu_1^2)=\widetilde O_d(n\mu_1)$.

Since there are $O(nm/(\mu_1\mu_2))$ block pairs, the total time is
\[
\widetilde O_d\!\left(
 nm\left(
 \frac1{\mu_2}
 +\frac{\mu_2}{\mu_1}
 +\frac{\mu_3}{\mu_2}
 +\frac1{\mu_3}
 +\frac1{\omega}
 +\frac{\omega\mu_2}{\mu_1}
 \right)
 + n\mu_1
\right).
\]
Choose $\mu_1=m^{4/5}$, $\mu_2=m^{2/5}$, $\mu_3=m^{1/5}$, and
$\omega=m^{1/5}$, rounded by constant factors so that
$\mu_3\mid\mu_2\mid\mu_1$.  Then $nm\,\mu_3/\mu_2$, $nm/\mu_3$,
$nm/\omega$, $nm\,\omega\mu_2/\mu_1$, and $n\mu_1$ all equal $nm^{4/5}$,
while $nm/\mu_2=nm^{3/5}$ and $nm\,\mu_2/\mu_1=nm^{3/5}$ are absorbed.
Therefore the discrete gap-decision running time is
$\widetilde O_{d}(nm^{4/5})$.

\begin{proof}[Proof of \Cref{thm:discrete-gap}]
On $\mathcal E_{\mathrm{disc}}$, the local cases and
\Cref{thm:disc-macro-transfer-recurrence} make every discrete \REACH{} call
complete and $5$-sound.  The block induction from
\Cref{subsec:discrete-block-dp} gives acceptance when
$\ddF(\tau,\sigma)\le\delta$ and rejection when
$\ddF(\tau,\sigma)>5\delta$.  By \Cref{lem:disc-macro-sampling},
$\mathcal E_{\mathrm{disc}}$ holds with high probability, and the explicit
sampling-failure cutoff bounds the running time even outside this event.
The time bound is $\widetilde O_{d}(nm^{4/5})$ by the calculation
above.
\end{proof}

\begin{theorem}[Discrete Fr\'echet approximation]
\label{thm:discrete-main}
For two polygonal curves $\tau$ and $\sigma$ in fixed dimension, with
$\abs{\tau}=n$ and $\abs{\sigma}=m\le n$, one can compute a randomized
$(5+\eps)$-approximation to their discrete Fr\'echet distance, with high
probability, in $\widetilde O_{d,\eps}(nm^{4/5})$ time.
\end{theorem}

\begin{proof}
A single run of the $O(\alpha)$-approximation of Bringmann and
Mulzer~\cite{bringmann2016approximability} with $\alpha=\min\{n/m^{4/5},\,n/\log n\}$ takes
$O(n\log n+n^2/\alpha)=O(n\log n+nm^{4/5})$ time and brackets
$\ddF(\tau,\sigma)$ within a factor $O(n)$; if it returns $0$, we are
done.  Binary search over a $(1+\eps/5)$-geometric grid of thresholds in
this bracket and call \Cref{thm:discrete-gap} at each step.  Initialize
$v$ one grid step below the bracket and $v'$ at its upper end; both
certificates follow from the bracket itself. It maintains a
rejected value $v$ (certifying $\ddF(\tau,\sigma)>v$) and an accepted
value $v'$ (certifying $\ddF(\tau,\sigma)\le 5v'$) until they are
grid-adjacent, at which point $5v'$ is a $(5+\eps)$-approximation. The
grid has $O(\eps^{-1}\log n)$ values, so there are
$O(\log(1/\eps)+\log\log n)$ gap-decision calls; choosing the sampling
constant so that a union bound over the calls preserves high probability,
the total running time is $\widetilde O_{d,\eps}(nm^{4/5})$.
\end{proof}

%% file: Discussion.tex
\section{Discussion}
\label{sec:discussion}

The approximation factor in this framework is governed by a single
quantity, the quality of the surrogates used in the bottom-to-top
transition.  Say that a surrogate $\pi$ for a piece $Z$ has
\emph{quality} $q$ if $\dF(\pi,Z)\le q\delta$ (discretely,
$\ddF(\pi,Z)\le q\delta$).  A witness subcurve $T[x,y]$ of the host
block with $\dF(T[x,y],Z)\le\delta$ is then within $(1+q)\delta$ of
$\pi$.  Transferring reachability through $\pi$ at threshold
$(1+q)\delta$ therefore loses no witness.  Conversely, every
transferred endpoint carries a certificate matching a subcurve of $T$
to $\pi$ at that threshold.  Converting the certificate back to $Z$
costs $q\delta$ once more.  The quality is thus paid twice.  A
surrogate search of quality $q$ decides the gap between $\delta$ and
roughly $(1+2q)\delta$.  Nothing else in the block dynamic program
depends on $q$, because the local cases stay at or below this
threshold whenever $q\ge1$.

This accounting places the known algorithms.  The bottom-to-top
surrogate construction of Cheng, Huang, and
Zhang~\cite{cheng2025constant} behaves as a quality-$3$ transfer in
this framework and yields their $(7+\eps)$ approximation factor.  Our
sampling-based search finds surrogates on the auxiliary curves of
\Cref{subsec:continuous-macro-surrogates,subsec:discrete-macro-surrogates}.
It realizes quality exactly $2$ in the discrete setting and quality
$2+\eta$ continuously, where $\eta\le\eps/6$ is the internal accuracy
parameter of \Cref{sec:continuous}.  The transfer structures operate
on the auxiliary curves directly, so no conversion back to the host
block is needed.  The transfer itself is exact in the discrete setting
(\Cref{thm:disc-dyadic-transfer}) and exact up to a factor $1+\eta$
continuously (\Cref{thm:canonical-transfer}).  The resulting constants
are exactly $5$ in \Cref{lem:disc-sequential-transfer} and at most $5+\eps$ in
\Cref{lem:aux-trans-sequential}.

Surrogates of quality $1$ always exist.  The witness subcurve itself
is one, and the free-start/free-end propagation of
\Cref{subsec:local-propagation} finds it.  Invoking it for every piece reconstructs the quadratic dynamic
program.  The obstacle is therefore not existence but total search
cost.

\begin{openproblem}\label{op:search}
Consider the discrete (respectively continuous) Fr\'echet distance.
Do there exist procedures, one for each host block $T$, that receive a
polygonal curve $Z$ and
\begin{enumerate}
\item either return a subcurve of $T$ certified to be within
Fr\'echet distance $(1+o(1))\delta$ of $Z$,
\item or report failure, which is permitted only when no subcurve of
$T$ is within Fr\'echet distance $\delta$ of $Z$,
\end{enumerate}
such that the total cost over all invocations made by the gap-decision
procedures is strongly subquadratic in $nm$?
\end{openproblem}

A positive answer, combined with the decision-to-optimization
conversions~\cite{colombe2021approximating,bringmann2016approximability},
gives strongly subquadratic $(3+\eps)$-approximation algorithms for
both variants.  This matches the conditional lower bound threshold
of~\cite{buchin2019seth}.  Partial progress pays immediately since any
quality $q<2$ gives a factor below $5$.  

Finally, the fine-macro samples of
\Cref{subsec:continuous-auxiliary-propagation,subsec:discrete-macro-auxtransfer}
are the only random choices, so modulo the surrogate search the
discrete gap-decision procedure is deterministic, and in this framework, the remaining randomness is isolated in the surrogate search.  Derandomizing
the search by exhaustively testing every fine macro costs
$\Theta(\mu_1\mu_2)$ per block pair and leads back to $\Theta(nm)$ in total.  What is missing is a
deterministic procedure that, within the sampling budget, either finds
a marked fine macro for a piece or certifies that fewer than $\omega$
are marked; a deterministic answer to \Cref{op:search} would settle
this as well.